%% file: template.tex
\RequirePackage{fix-cm}
\documentclass[twocolumn]{svjour3}          
\smartqed  
\usepackage{graphicx}
\input{preamble}
\begin{document}
\begin{sloppypar}

\title{Efficient Execution of SPARQL Queries with OPTIONAL and UNION Expressions
}


\author{Lei Zou \and Yue Pang \and M. Tamer Özsu \and Jiaqi Chen}


\institute{Lei Zou \at
              Peking University, Beijing, China\\\email{zoulei@pku.edu.cn}
           \and
           Yue Pang \at Peking University, Beijing, China\\\email{michelle.py@pku.edu.cn}
           \and
           M. Tamer Özsu \at University of Waterloo, Waterloo, Canada\\\email{tamer.ozsu@uwaterloo.ca}
           \and
           Jiaqi Chen \at Peking University, Beijing, China\\\email{chenjiaqi93@pku.edu.cn}}

\date{Received: date / Accepted: date}

\maketitle

\begin{abstract}
The proliferation of RDF datasets has resulted in studies focusing on optimizing SPARQL query processing. Most existing work focuses on basic graph patterns (BGPs) and ignores other vital operators in SPARQL, such as \texttt{UNION} and \texttt{OPTIONAL}. SPARQL queries with these operators, which we abbreviate as SPARQL-UO, pose serious query plan generation challenges. In this paper, we propose techniques for executing SPARQL-UO queries using BGP execution as a building block, based on a novel \underline{B}GP-based \underline{E}valuation (BE)-Tree representation of query plans. On top of this, we propose a series of \emph{cost-driven BE-tree transformations} to generate more efficient plans by reducing the search space and intermediate result sizes, and a \emph{candidate pruning} technique that further enhances efficiency at query time. Experiments confirm that our method outperforms the state-of-the-art by orders of magnitude.
\end{abstract}

\input{introduction}

\input{preliminary}

\input{baseline}
\input{method}
\input{experiment}
\input{conclusion}


\bibliographystyle{spmpsci}      
\bibliography{sample-base}   

\clearpage
\input{appendix}

\end{sloppypar}
\end{document}

%% file: preamble.tex
\usepackage{graphicx}
\usepackage{balance}  

\usepackage{ifsym}
\usepackage{textcomp}
\usepackage{makecell,multirow,diagbox}
\usepackage[linesnumbered,ruled,noend]{algorithm2e}

\newcommand{\nop}[1]{}

\newcommand{\Paragraph}[1]{~\vspace*{-0.0\baselineskip}\\{\bf #1}}
\usepackage{setspace}

\usepackage{caption}
\usepackage{subcaption}

\usepackage{listings}
\usepackage[charter]{mathdesign}
\usepackage{amsmath}
\usepackage{hyperref}
\usepackage{blindtext}

%% file: introduction.tex
\vspace{-0.10in}
\section{Introduction}
The proliferation of knowledge graphs has generated many RDF (Resource Description Framework) data management problems. RDF is the de-facto data model for knowledge graphs, where each edge is a triple of $\langle$subject, predicate, object$\rangle$. SPARQL has been the focus of a significant body of research as the standard language for accessing RDF datasets. Most of the existing work focus on basic graph pattern (BGP) execution \cite{Neumann2009,DBLP:journals/pvldb/ZouMCOZ11,DBLP:journals/pvldb/YuanLWJZL13}, which is the basic building block of SPARQL. On the other hand, how to execute and optimize queries containing operators on graph patterns, such as \texttt{UNION} and \texttt{OPTIONAL}, has received much less attention.


\texttt{UNION} and \texttt{OPTIONAL} expressions are essential in SPARQL grammar. RDF is a semi-structured data model that does not enforce the underlying data to adhere to a predefined schema, which provides flexibility in integrating diverse sources of RDF data, but leads to challenges when issuing queries since the same information can be represented in many ways in RDF graphs. The \texttt{UNION} operator is crucial in this case since it groups diversely expressed information. For example, in DBpedia \cite{DBLP:journals/semweb/LehmannIJJKMHMK15}, an open-domain knowledge graph extracted from Wikipedia, persons' names are represented using the predicate $\langle$\texttt{foaf:name}$\rangle$ or $\langle$\texttt{rdfs:label}$\rangle$. Thus, to fully retrieve all the names of a group of persons (e.g., Presidents of the United States), it is necessary to use the \texttt{UNION} operator (Figure \ref{fig:bgpeupb}(a)).

\begin{figure*}
	\centering
	\includegraphics[scale=0.5]{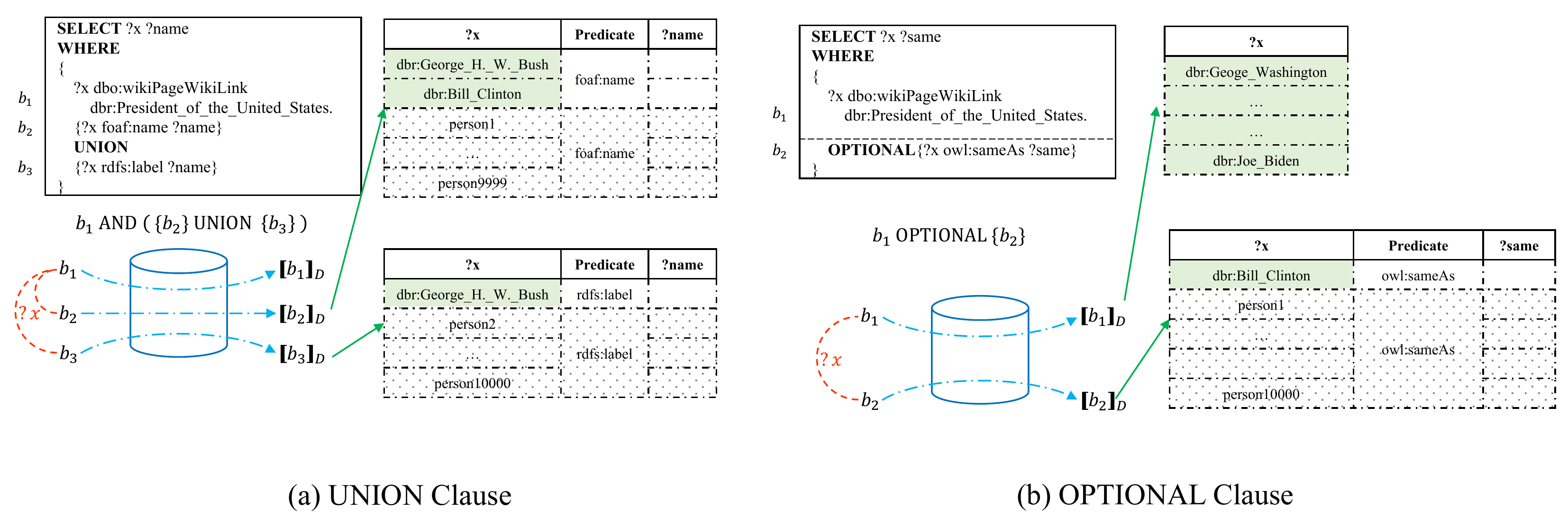}
	\vspace{-0.1in}
	\caption{An Example Query with a \texttt{UNION} and  \texttt{OPTIONAL} Clause}
	\label{fig:bgpeupb}
	\vspace{-0.1in}
\end{figure*}

In addition to the diversity of representation, incompleteness is another feature of RDF datasets. Specifically, an entity may lack some attributes or relationships of most other similar entities (which are most likely to be stored in the same table in a relational database). The {OPTIONAL} operator is crucial in this case since it allows attaching some attributes or relations as optional information. For example, the \texttt{OPTIONAL} query in Figure \ref{fig:bgpeupb}(b) fetches all the presidents of the United States, along with other references to them that are not on the same Wikipedia page (through the predicate \texttt{owl:sameAs}). Since not every president has multiple references in the database, the triple with the predicate \texttt{owl:sameAs} is enclosed in an \texttt{OPTIONAL} expression, so those presidents without alternative references are still retained in the results.

\nop{
\begin{figure*}
	\centering
	\includegraphics[scale=1.0]{figure/bgpeopb-new-new.png}
	\vspace{-0.1in}
	\caption{An Example Query with an \texttt{OPTIONAL} Clause}
	\label{fig:bgpeopb}
	\vspace{-0.1in}
\end{figure*}
}

\texttt{UNION} and \texttt{OPTIONAL} expressions are widely used in real-world SPARQL workloads and are part of the SPARQL 1.1 specification. Recent empirical studies \cite{bonifati2020analytical} show that \texttt{UNION} and \texttt{OPTIONAL} expressions occur in $25.10\%$ and $31.72\%$ of the valid queries from real SPARQL query logs across a diverse range of endpoints, respectively. In this paper, we address the efficient execution of SPARQL queries with \texttt{UNION} and \texttt{OPTIONAL} expressions, which we abbreviate as SPARQL-UO queries.
\Paragraph{Our Solution.} Since BGP has been well studied, it is desirable to build SPARQL-UO query optimization on a well-performing BGP engine. Therefore, we first propose a BGP-based query evaluation scheme. Specifically, we propose a \emph{B}GP-based \emph{E}valuation (BE)-tree (Definition \ref{def:betree}) for the evaluation plan of SPARQL-UO queries. However, if the BE-tree is evaluated as it is, some BGPs may generate large intermediate results. Therefore, we propose a BE-tree transformation method to generate a more efficient query plan.

We introduce two types of transformations, \emph{merge} and \emph{inject}, that target \texttt{UNION} and \texttt{OPTIONAL} operators, respectively. These transformations expose opportunities for reducing the cost during the evaluation of BGPs, \texttt{UNION} and \texttt{OPTIONAL} operators while preserving query semantics. Since there are many different ways to transform a BE-tree,
we
devise a cost model that accounts for the cost of evaluating both BGPs and these operators, and choose the transformation that most reduces the cost. Because of the vast space of possible transformations,
we propose a greedy strategy to determine the transformation step-by-step. The transformed BE-tree is then evaluated by the BGP-based scheme, enhanced by the query-time optimization called \emph{candidate pruning}, which prunes the search space of BGP evaluation on-the-fly whenever possible.

To summarize, we make the following contributions: 
\begin{enumerate}
\item We propose a novel BE-tree representation for the evaluation plan of a SPARQL-UO query and design two BE-tree transformation primitives, \emph{merge} and \emph{inject}, to generate more efficient SPARQL-UO query plans.    
\item We propose a cost model for SPARQL-UO queries and a cost-driven BE-tree transformation algorithm.
\item We design a query-time optimization called \emph{candidate pruning} that augments the BGP-based query evaluation scheme by pruning the search space.
\item We conduct experiments on large-scale real and synthetic RDF datasets, which shows that our method outperforms existing techniques by orders of magnitude.
\end{enumerate}

The remainder of the paper is organized as follows. A brief review of related work is given in Section \ref{sec:related}. Some necessary preliminary information is laid out in Section \ref{sec:preliminary}. Section \ref{sec:bgp} presents the BE-tree plan representation and transformations, as well as the query evaluation scheme. The cost-driven plan selection algorithm is proposed in Section \ref{sec:plan_selection}, and the query-time optimization is presented in \ref{sec:cand_pruning}. We experimentally evaluate our method in Section \ref{Sec:experiments} and conclude the paper in Section \ref{sec:conclusion}.

\section{Related Work}\label{sec:related}

Although the optimization of SPARQL queries has been extensively studied, most of the focus has been on evaluating BGPs \cite{DBLP:conf/sigmod/Atre15}, including graph-based approaches and relational approaches. Graph-based approaches include works on effective index strategies (e.g., gStore \cite{DBLP:journals/pvldb/ZouMCOZ11}) and join order optimization (e.g.,WCOJ \cite{10.1007/978-3-030-30793-6_15}). In contrast, relational approaches rely on a relational DBMS and consider RDF graphs as three-column tables or other complex table organizations \cite{vldb07_Abadi:2007,DBLP:conf/sigmod/BorneaDKSDUB13}. Processing SPARQL queries is then mapped to its relational counterparts, as done in Apache Jena \cite{Wilkinson:Jena2} and Virtuoso \cite{virtuosourl}. 
Relational BGP optimization approaches focus primarily on efficient data organization (e.g., property table \cite{Wilkinson:Jena2}, vertical partitioning \cite{DBLP:journals/vldb/AbadiMMH09} and single table exhaustive indexing \cite{Neumann2009}). However, these BGP optimization techniques
cannot optimize \texttt{UNION} and \texttt{OPTIONAL} since the semantics of these operators are fundamentally different from joins. As explained in later sections, our solution relies on BGP evaluation as a basic building block, and our proposed optimization techniques operate on a higher level than BGP evaluation techniques.


The existing research on SPARQL with \texttt{UNION} and \texttt{OPTIONAL} operators is primarily theoretical, studying their semantics and complexity \cite{DBLP:journals/tods/PerezAG09}. For example, Letelier et al. \cite{10.1145/2500130} propose a WDPT (well-designed pattern tree), which focuses on the analysis of containment and equivalence of a class of SPARQL graph patterns called \emph{well-designed patterns} and identifying the tractable components of their evaluation. However, no work has yet considered SPARQL-UO query optimization from the systems perspective, i.e., how to design an efficient SPARQL query processor to evaluate SPARQL-UO queries. To the best of our knowledge, LBR \cite{DBLP:conf/sigmod/Atre15} is the only work that considers \texttt{OPTIONAL} query optimization. It designs a new data structure GoSN, which is reminiscent of WDPT but focuses on the practical aspects of \texttt{OPTIONAL} pattern evaluation. Concretely, it proposes a query rewriting technique to reduce intermediate results of left-outer joins, the join semantics represented by \texttt{OPTIONAL}. To remove inconsistent variable bindings, LBR uses the \emph{nullification} and \emph{best-match} techniques previously studied in SQL left-outer joins \cite{DBLP:conf/sigmod/RaoPZ04}. LBR also proposes a semijoin strategy to prune candidates, extending the operator for the minimality of acyclic inner joins \cite{DBLP:journals/jacm/BernsteinC81}. However, it follows an execution strategy of two-pass semijoin scans following the graph of join variables, which introduces additional overhead during query execution.
In this paper, we propose a more comprehensive approach that deals with \texttt{UNION} and \texttt{OPTIONAL}.
Experiments also demonstrate that our techniques significantly outperform LBR on \texttt{OPTIONAL} queries. 





Note that our techniques to optimize \texttt{UNION} expressions can also be applied to conjunctive relational queries with unions due to their semantic similarity. In fact, all SPARQL-UO queries can be equivalently mapped to SQL, but the mapping of \texttt{OPTIONAL} expressions involves sub-selects in SQL, so our techniques cannot be applied without major adaptations \cite{prud'hommeaux_bertails_2008,chebotko2009semantics}. We are aware of a recent demonstration \cite{al2017optimizing} that proposes a \emph{join pushing} technique on conjunctive queries with unions, which pushes the join condition into the unioned sets if a cost model deems it more efficient. This is principally similar to our approach when applied to relational queries, but no description of the employed cost model is provided, which renders further comparison impossible.


%% file: preliminary.tex
\vspace{-0.1in}
\section{Preliminaries}\label{sec:preliminary}
\vspace{-0.1in}
\Paragraph{RDF Dataset.} Table \ref{tab:rdf} is an example RDF dataset (defined as follows) containing seven triples.



\begin{definition} [RDF dataset] Let pairwise disjoint infinite sets $I$, $B$, and $L$ denote IRI, blank nodes and literals, respectively. An RDF dataset $D$ is a collection of triples $D = \{ t_1, t_2, ..., t_{|D|} \}$, where each triple is a three-tuple  $t = \langle subject,property,object \rangle \in (I \cup B)  \times I \times (I\cup B \cup L)$.
\end{definition} 
\Paragraph{SPARQL Query---Syntax.} Assume that there is an infinite set $V$ representing the variables that appear in the query. All variables differ from IRIs and literals by leading with a question mark (\texttt{?}), so  the set $V$ is disjoint with $I$ and $L$. This work focuses on \texttt{SELECT} queries, which retrieve results by matching the graph pattern in the query with the dataset. We note that SPARQL provides other query forms for updating the database, constructing RDF datasets, asking whether a graph pattern exists in the database and describing resources, etc., that are beyond the scope of our consideration. A \texttt{SELECT} query is of the form ``\texttt{SELECT} ${v_1}$ ${v_2}$ ... ${v_k}$ \texttt{WHERE} \{...\}'', in which the \texttt{SELECT} clause represents the query header, and the \texttt{WHERE} clause represents the query body (Figure \ref{fig:sparql}(a)). The \texttt{SELECT} clause determines the projection variables that need to appear in the query results, and the \texttt{WHERE} clause gives the group graph pattern that needs to be matched over the RDF dataset, which may consists of many other types of graph patterns, defined as follows.




\begin{table*}[!htbp]
	\centering
	\caption{An example RDF dataset}
	\scriptsize
	\begin{tabular}{|c|c|l|}
		\hline
		Subject & Predicate & \multicolumn{1}{c|}{Object} \\ \hline
		\texttt{dbr:George\_W.\_Bush} & \texttt{foaf:name} & ``George Walker Bush"@en \\ \hline
		\texttt{dbr:George\_W.\_Bush} & \texttt{rdfs:label} & ``George W. Bush"@en \\ \hline
		\texttt{dbr:George\_W.\_Bush} & \texttt{dbo:wikiPageWikiLink} & \texttt{dbr:President\_of\_the\_United\_States} \\ \hline
		\texttt{dbr:Bill\_Clinton} & \texttt{foaf:name} & ``Bill Clinton"@en \\ \hline
		\texttt{dbr:Bill\_Clinton} & \texttt{dbo:wikiPageWikiLink} & \texttt{dbr:President\_of\_the\_United\_States} \\ \hline
		\texttt{dbr:Bill\_Clinton} & \texttt{dbp:birthDate} & ``1946-08-19"\textasciicircum \textasciicircum \texttt{xsd:date} \\ \hline
		\texttt{dbr:Bill\_Clinton} & \texttt{owl:sameAs} & \texttt{fbp:Clinton\_William\_Jefferson\_1946-} \\ \hline
	\end{tabular}
	\label{tab:rdf}
\end{table*}





\begin{definition}[Triple Pattern]\label{def:triple_pattern}
A triple $t \in (V \cup I) \times (V \cup I) \times (V \cup I \cup L)$ is a \emph{triple pattern}.
\end{definition}

Basic graph patterns (BGPs) are composed of triple patterns. To give a formal definition of BGPs, we need to first introduce the notion of \emph{coalescability}.

\begin{definition}[Coalescable triple patterns] 
\label{def:coalesce-triple}
We say that the triple patterns $t_1 = \langle s_1,p_1,o_1 \rangle$ and $t_2 = \langle s_2,p_2,o_2 \rangle$ are \underline{coalescable} if and only if $\{s_1$, $o_1\}$ and $\{s_2$, $o_2\}$ share at least one common variable. 
\end{definition}
\nop{
\begin{itemize}
    \item $s_1$ and $s_2$ are variables, $s_1 = s_2$;
    \item $s_1$ and $o_2$ are variables, $s_1 = o_2$;
    \item $o_1$ and $s_2$ are variables, $o_1 = s_2$;
    \item $o_1$ and $o_2$ are variables, $o_1 = o_2$.
\end{itemize}
}

Intuitively, two triple patterns are coalescable if they have common variables at the subject or object positions. Since a BGP is composed of triple patterns, we can extend coalescability to BGPs, where we require some of their constituent triple patterns to be coalescable.


\begin{definition}[Coalescable BGPs] 
\label{def:coalesce-bgp}
We say that the BGPs $b_1$ and $b_2$ are \underline{coalescable} if there exist $t_{i_1} \in b_1$ and $t_{i_2} \in b_2$ such that $t_1$ and $t_2$ are coalescable triple patterns.
\end{definition}


\begin{definition}[Basic Graph Pattern (BGP)]\label{def:bgp}
A BGP is recursively defined as follows:
\begin{enumerate}
\item  A triple pattern $t$ is a BGP;
\item if $P_1$ and $P_2$ are coalescable BGPs, $P_1$ \texttt{AND} $P_2$ is also a BGP.
\end{enumerate}
\end{definition}

\begin{definition}[Graph Pattern, Group Graph Pattern] \label{def:grouppattern}

A \underline{graph pattern} is recursively defined as follows:
\begin{enumerate}
	\item if $P$ is a BGP, $P$ is a graph pattern;
	\item if $P$ is a group graph pattern (defined below), $P$ is a graph pattern;
	\item if $P_1$ and $P_2$ are both graph patterns, $P_1$ \texttt{AND} $P_2$ is also a graph pattern;
            \item if $P_1$ and $P_2$ are both graph patterns,  $\{P_1 \}$ \texttt{UNION} $\{P_2\}$, $P_1$ \texttt{OPTIONAL}
$\{P_2\}$ are both graph patterns. Note that $\{P_i\}$ denotes a \textit{group graph pattern} (defined below);
\end{enumerate}

A \emph{group graph pattern} $P$ is recursively defined as follows:
\begin{enumerate}
            \item If $P$ is a graph pattern, $\{P\}$ is a group graph pattern.
\end{enumerate}
\end{definition}

Figure \ref{fig:sparql} is an example SPARQL query with six triple patterns (${t_{1..6}}$) and \texttt{UNION} and \texttt{OPTIONAL} operators.


\begin{figure}
	\centering
	\includegraphics[scale=0.55]{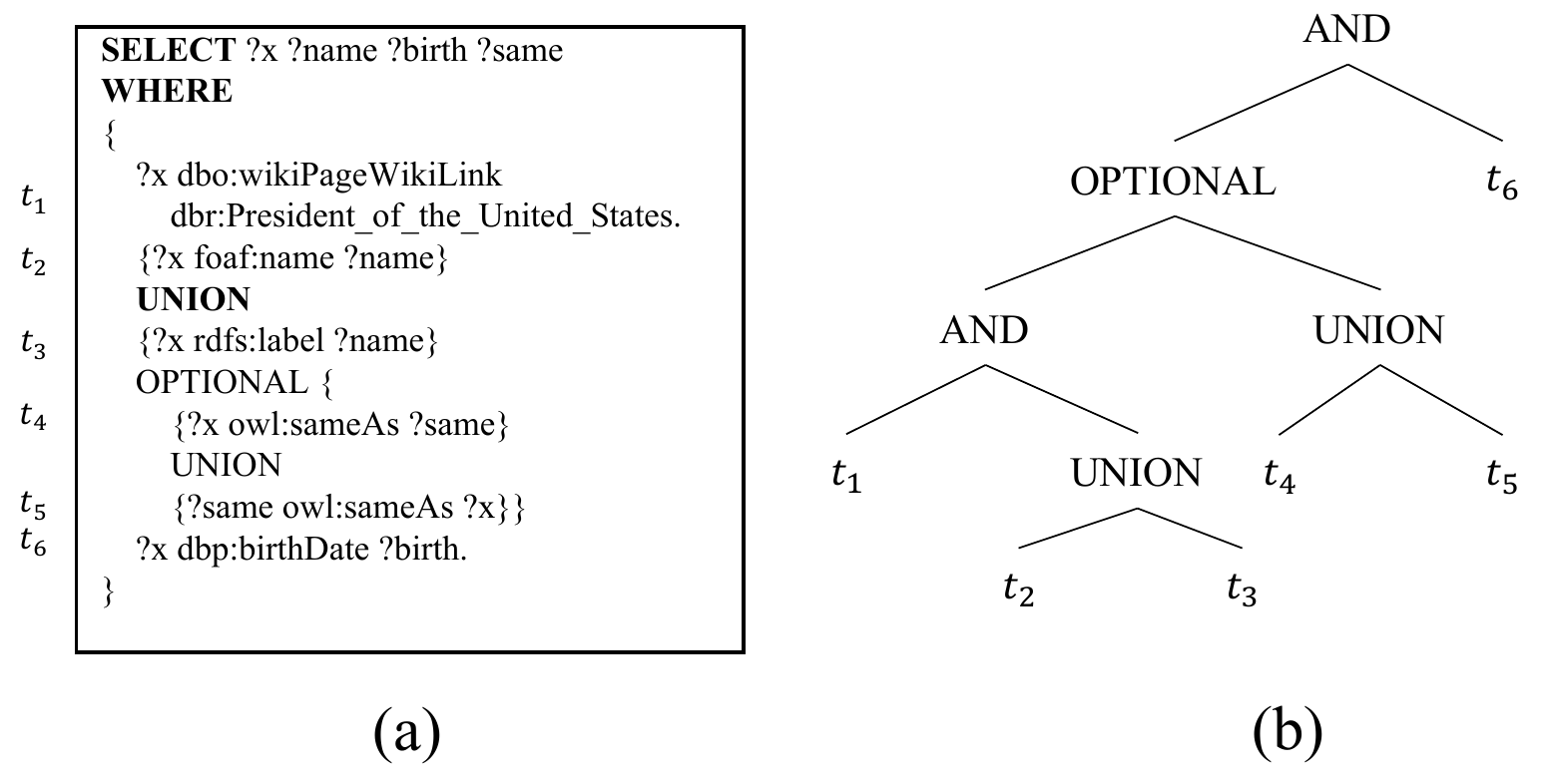}
	\vspace{-0.1in}
	\caption{(a) An example SPARQL query and (b) Binary Tree Expression}
	\label{fig:sparql}
	\vspace{-0.1in}
\end{figure}

\nop{
\begin{figure}
	\centering
	\includegraphics[scale=0.6]{figure/semexpt.pdf}
	\vspace{-0.1in}
	\caption{Binary tree expression of Graph Pattern}
	\label{fig:semexpt}
	\vspace{-0.1in}
\end{figure}
}
\Paragraph{SPARQL Query---Semantics.} The semantics of any graph pattern can be uniquely determined, since the \texttt{OPTIONAL} clause is left-associative and the priority of operators is defined as $\{\}$ $\prec$ \texttt{UNION} $\prec$ \texttt{AND} $\prec$ \texttt{OPTIONAL}.



A graph pattern can then be converted to an expression containing triple patterns, built-in conditions, and binary operators \texttt{AND}, \texttt{UNION} and \texttt{OPTIONAL}, which accept two graph patterns as their operands. Such an expression can be equivalently represented by a binary tree, where each leaf node represents a triple pattern and each internal node represents a binary operator. Figure \ref{fig:sparql}(b) shows such a binary tree expression of the outermost group graph pattern of the query in Figure \ref{fig:sparql}(a).

A graph pattern $P$ is matched on an RDF dataset $D$ (denoted by $[\![P]\!]_D$) to produce a bag (i.e., multi-set) of mappings $\{{\mu_1}, {\mu_2}, \ ... \ , \ {\mu_n}\}$, which may contain duplicate mappings. A mapping $\mu : V \mapsto U$ is a partial function from $V$ to $(I \cup L)$, where $V$ represents the variables that appear in the query, and $I$ and $L$ denote the sets of IRI and literals, respectively. The set of variables appearing in mapping $\mu$ is denoted by $dom(\mu)$. The two mappings ${\mu_1}$ and ${\mu_2}$ are defined to be \emph{compatible} (denoted by ${\mu_1} \sim {\mu_2}$) if and only if for all variables $v \in dom({\mu_1}) \ \cap \ dom({\mu_2})$ satisfying ${\mu_1}(v) = {\mu_2}(v)$. Intuitively, this means that the common variables of ${\mu_1}$ and ${\mu_2}$ are mapped to the same values. In the case where ${\mu_1}$ and ${\mu_2}$ are compatible, ${\mu_1} \cup {\mu_2}$ is also a mapping. If the two mappings ${\mu_1}$ and ${\mu_2}$ are \emph{incompatible}, we denote the case as ${\mu_1} \nsim {\mu_2}$.

We denote two bags of mappings by ${\Omega_1}$ and ${\Omega_2}$, and define several operators on bags as follows:
\begin{enumerate}
	\item ${\Omega_1} \Join {\Omega_2} = \{ {\mu_1} \cup {\mu_2} \ | \ {\mu_1} \in {\Omega_1} \wedge {\mu_2} \in {\Omega_2} \wedge {\mu_1} \sim {\mu_2} \}$.
	\item ${\Omega_1} \cup_{bag} {\Omega_2} =  \{ {\mu_1} \ | \ {\mu_1} \in {\Omega_1} \} \bigcup_{bag} \{ {\mu_2} \ | \ {\mu_2} \in {\Omega_2} \}$.
	\item ${\Omega_1} \setminus {\Omega_2} = \{{\mu_1} \in {\Omega_1}\ | \ \forall {\mu_2} \in {\Omega_2} \ : \ {\mu_1} \nsim {\mu_2} \}$.
	\item ${\Omega_1} \ {\tiny \textbf{\textifsym{d|><|}}} \ {\Omega_2} = ({\Omega_1} \Join {\Omega_1}) \bigcup_{bag} ({\Omega_1} \setminus {\Omega_1})$
\end{enumerate}

Note that the operators above all preserve duplicate elements, as they follow the bag semantics.

\begin{definition} [Evaluation of Graph Patterns on an RDF dataset]\label{def:eval}
The evaluation of graph patterns $P$ on an RDF dataset $D$ (denoted by $[\![P]\!]_D$) is recursively defined as follows:
\begin{enumerate}
	\item If $P$ is a triple pattern $t$, $[\![P]\!]_D = \{ \mu \ | \ var(t)=dom(\mu) \ \wedge \ \mu(t) \in D \}$ ($var(t)$ represents all variables occurring in $t$, and $\mu(t)$ mean that all variables appearing in $t$ are replaced by $\mu$).
	\item If $P = \{P_1\}$, $[\![P]\!]_D = [\![P_1]\!]_D$.
	\item If $P = (P_1 \ \texttt{AND} \ P_2)$, $[\![P]\!]_D = [\![P_1]\!]_D \Join [\![P_2]\!]_D$.
	\item If $P = (P_1 \ \texttt{UNION} \ P_2)$, $[\![P]\!]_D = [\![P_1]\!]_D \bigcup_{bag} [\![P_2]\!]_D$.
	\item If $P = (P_1 \ \texttt{OPTIONAL} \ P_2)$, $[\![P]\!]_D = [\![P_1]\!]_D \ {\tiny \textbf{\textifsym{d|><|}}} \ [\![P_2]\!]_D$. We say that ${P_1}$ is a \texttt{OPTIONAL}-left graph pattern, and ${P_2}$ is a \texttt{OPTIONAL}-right graph pattern.
\end{enumerate}
\end{definition}

%% file: baseline.tex
\section{Plan Representation: BGP-based Evaluation Tree}\label{sec:bgp}

The most straightforward approach for evaluating a graph pattern $P$ is to employ a bottom-up strategy on the binary tree representation. In each step, we either evaluate a triple pattern, or perform a binary operator (\texttt{AND}, \texttt{UNION}, or \texttt{OPTIONAL}).
This binary-tree-based evaluation  strictly follows the SPARQL semantics discussed in Section \ref{sec:preliminary}, but it has a number of inherent performance limitations due to the large number of intermediate results generated for each triple pattern at the leaf nodes of the binary tree expression. To illustrate this, consider the simple SPARQL query in Figure \ref{fig:semetpb}. Note that the outermost group graph pattern of this query only contains a BGP. Following the binary tree expression-based method, we first need to obtain $[\![t_1]\!]_D$ and  $[\![t_2]\!]_D$. Obviously, the triple pattern $t_2$ will generate a large number of intermediate results, since most persons in the database have their birth dates as an attribute.

\begin{figure*}[ht]
	\centering
	\includegraphics[scale=1.4]{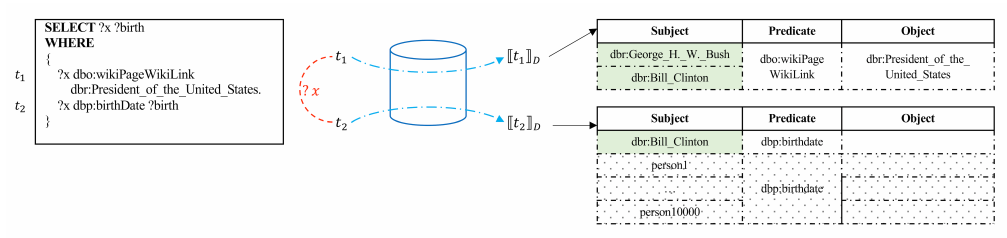}
	\vspace{-0.1in}
	\caption{Inefficiency of binary-tree-based query evaluation}
	\label{fig:semetpb}
	\vspace{-0.1in}
\end{figure*}

It is evidently more desirable to use BGP evaluation as the basic building block for executing SPARQL queries, employing an optimized BGP query evaluation method such as those used in RDF-3x \cite{Neumann2009}, SW-store \cite{vldb07_Abadi:2007}, gStore \cite{DBLP:journals/pvldb/ZouMCOZ11} and Jena~\cite{Wilkinson:Jena2}. Therefore, in our approach, we design a \underline{B}GP-based \underline{E}valuation Tree (BE-tree) to represent a SPARQL query evaluation plan.


\subsection{BE-Tree Structure}

\nop{
\begin{definition}[\underline{B}GP-based \underline{E}valuation Tree (BE-tree)] 
\label{def:betree}
Given a group graph pattern $Q$, its corresponding BE-tree $T(Q)$ is recursively defined as follows:
\begin{itemize}
    \item The root of $T(Q)$ is a group graph pattern node representing query $Q$;
    \item A \underline{group graph pattern node} represents a group graph pattern (Definition \ref{def:grouppattern}). It has one or more child nodes, which may be a group graph pattern node, a BGP node, a \texttt{UNION} node, an \texttt{OPTIONAL} node, or a \texttt{FILTER} node;
    \item A \underline{BGP node} represents a BGP (Definition \ref{def:bgp}), and must be a leaf node;
    \item A \underline{\texttt{UNION} node} represents the \texttt{UNION} expression that links two or more group graph patterns, called \texttt{UNION}'ed group graph patterns. It has two or more child nodes, which are all group graph pattern nodes;
    \item An \underline{\texttt{OPTIONAL} node} represents the \texttt{OPTIONAL} expression that links the graph patterns to its left and the adjacent group graph pattern to its right, called the \texttt{OPTIONAL}-right group graph pattern. It has exactly one child node, which is a group graph pattern node representing the adjacent group graph pattern to its right;
    \item A \underline{\texttt{FILTER} node} represents the \texttt{FILTER} expression and its accompanying built-in condition, and must be a leaf node.
\end{itemize}
\end{definition}
}

\begin{definition}[\underline{B}GP-based \underline{E}valuation Tree (BE-tree)] 
\label{def:betree}
Given a group graph pattern $Q$, its corresponding BE-tree $T(Q)$ is recursively defined as follows:
\begin{itemize}
    \item The root of $T(Q)$ is a \underline{group graph pattern node} (Definition \ref{def:grouppattern}) representing the query $Q$;
    \item An internal node of $T(Q)$ can be one of \{\texttt{UNION}, \texttt{OPTIONAL}, group graph pattern\} nodes:
    \begin{itemize}
        \item A \underline{\texttt{UNION} node} represents the \texttt{UNION} expression that links two or more group graph patterns, called \texttt{UNION}'ed group graph patterns. It has two or more child nodes, which are all group graph pattern nodes;
        \item An \underline{\texttt{OPTIONAL} node} represents the \texttt{OPTIONAL} expression that links \texttt{OPTIONAL}-left and \texttt{OPTIONAL}-right graph patterns. It has exactly one child node: the \texttt{OPTIONAL}-right graph pattern, which is a group graph pattern node;
    \end{itemize}
        \item A leaf node of $T(Q)$ is a \texttt{BGP} node (Definition \ref{def:bgp}).
        \nop{\begin{itemize}
            \item A \underline{BGP node} represents a BGP (Definition \ref{def:bgp});
            \item A \underline{\texttt{FILTER} node} represents the \texttt{FILTER} expression and its accompanying built-in condition.
        \end{itemize}}
\end{itemize}
\end{definition}

According to the above definition, each leaf node in a BE-tree corresponds to a BGP, and each internal node corresponds to a group graph pattern, a \texttt{UNION} expression, or an \texttt{OPTIONAL} expression. 
Figure \ref{fig:hrgp} shows the general structure of a BE-tree. The edge labels indicate how many child nodes of this type are permitted to occur: $k$ indicates that exactly $k$ such child nodes must occur, and $k..*$ indicates that $k$ or more such child nodes can occur. For convenience, we call a group of sibling nodes a \emph{level} of nodes. 

\begin{figure}[ht]
	\centering
	\includegraphics[scale=0.325]{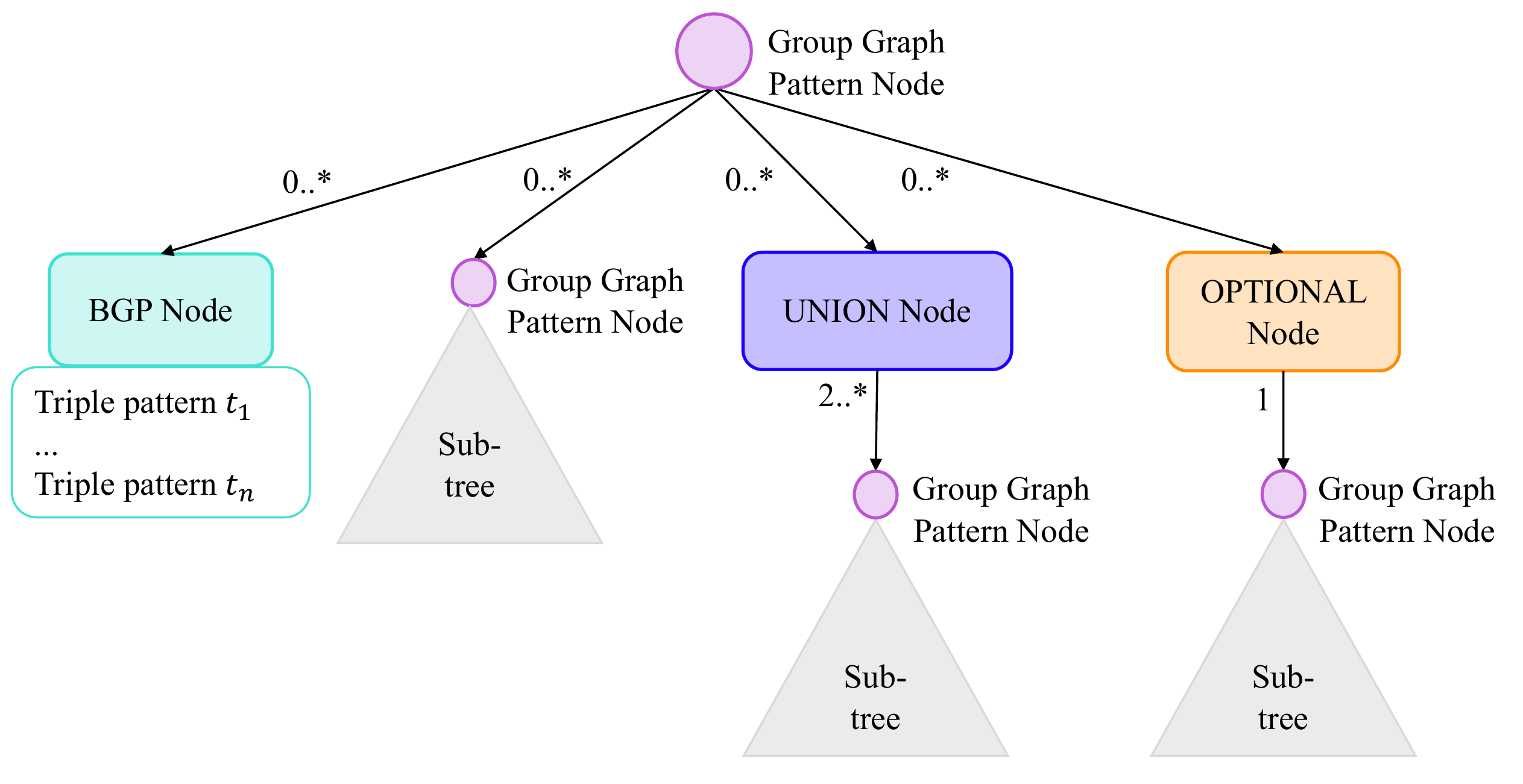}
	\vspace{-0.1in}
	\caption{Hierarchical structure of the BE-tree}
	\label{fig:hrgp}
	\vspace{-0.1in}
\end{figure}

It is straightforward to construct a BE-tree from a SPARQL query. Joins between graph patterns are implicitly expressed in the BE-tree as the sibling relation between nodes. Therefore, we first initiate a group graph pattern node as the root, denoting the outermost group graph pattern in the query. Then we put each joined graph pattern within the outermost group graph pattern as the root's children in the original order. For each nested group graph pattern, we consider them in turn as the root of a subtree, and recursively execute the aforementioned process.

Note that such a construction procedure generates triple pattern nodes, which are not identified in Definition \ref{def:betree}. In order to eliminate them, we coalesce sibling triple pattern nodes into \emph{maximal} BGP nodes, in that no further coalescing can be performed (The coalescability of triple patterns and BGPs is defined in Definitions \ref{def:coalesce-triple} and \ref{def:coalesce-bgp}). We place BGP nodes where its constituent leftmost triple pattern originally resides. It is evident that there is a one-to-one mapping between SPARQL queries and BE-trees by this construction process.

As a concrete example, the BE-tree of the query in Figure \ref{fig:sparql}(a) is given below (Figure \ref{fig:betree-example}). Note that the triple patterns $t_1$ and $t_6$ are coalesced to form a BGP node; no other triple patterns cannot be coalesced, and thus form individual BGP nodes on their own.

\begin{figure}[ht]
	\centering
	\includegraphics[scale=0.35]{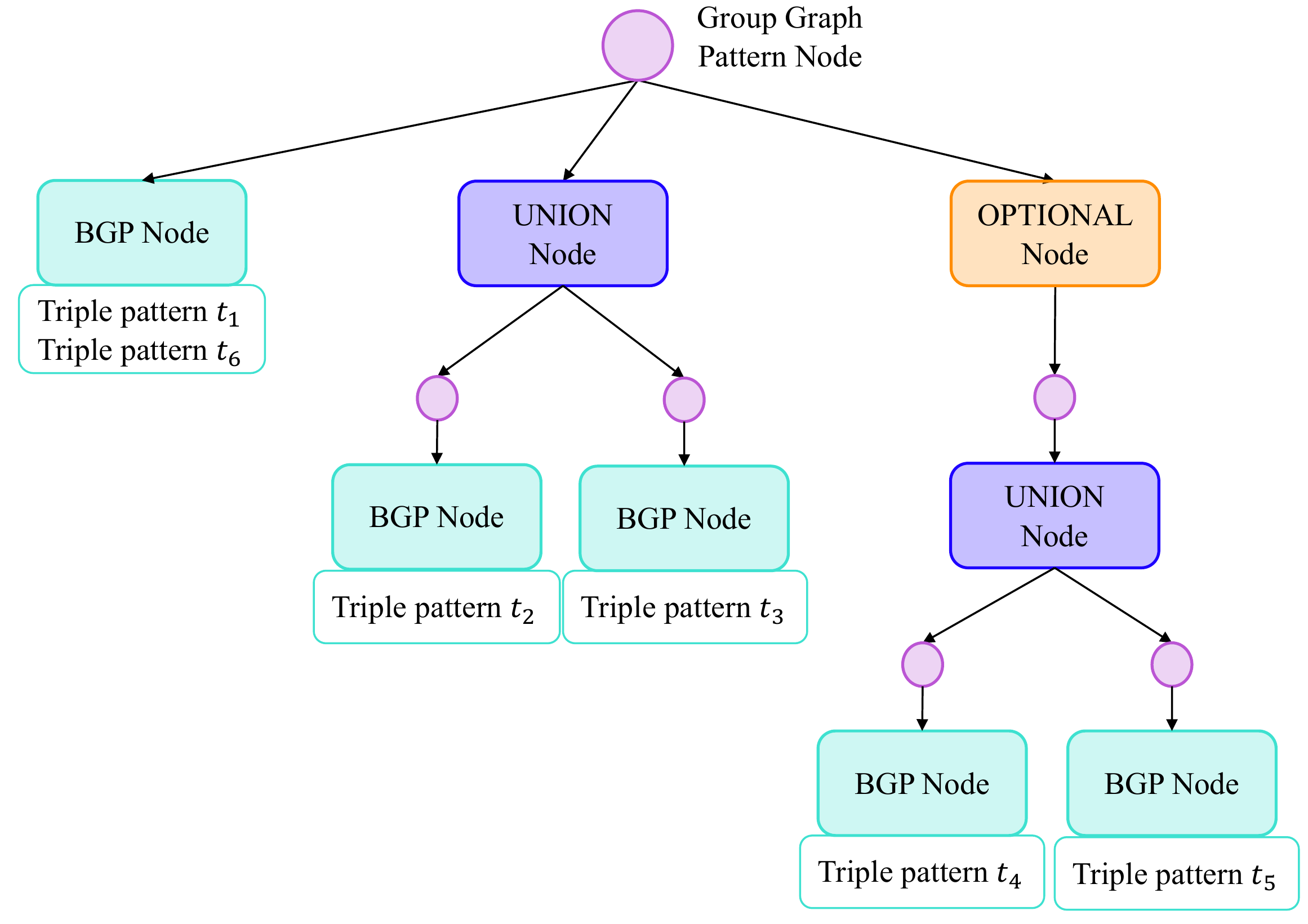}
	\vspace{-0.1in}
	\caption{An example BE-tree}
	\label{fig:betree-example}
	\vspace{-0.1in}
\end{figure}

\begin{algorithm}\small
	\caption{\small BGP-based query evaluation}
	\label{alg:bgpe}
	\SetKwInOut{KwIn}{Input}
	\KwIn{RDF dataset $D$, BE-tree $T(Q)$}
	\SetKwInOut{KwOut}{Output}
	\KwOut{${\left[\!\left[ {Q} \right]\!\right]_D}$}
	
	\SetKwFunction{FMain}{BGPBasedEvaluation}
	\SetKwFunction{FBGP}{EvaluateBGP}
	\SetKwProg{Fn}{Function}{:}{}
	\Fn{\FMain{$D, T(Q)$}}
	{
		$r \leftarrow \emptyset$\;
		Let $root$ be the root of $T(Q)$ \;
		\ForEach{child node $e_i$ of $root$}
		{
			\If {$e_i$ is a group graph pattern node}
			{
				\If {$r = \emptyset$}
				{
					$r \leftarrow \FMain(D, T(e_i))$\;
				}
				\Else
				{
					$r \leftarrow r \ \Join \ \FMain(D, T(e_i))$\;
				}
			}
			\ElseIf {$e_i$ is a BGP node}
			{
				$r \leftarrow r \ \Join \ \FBGP(D, e_i)$\;
			}
			\ElseIf {$e_i$ is a \texttt{UNION} node}
			{
				$u \leftarrow \emptyset$\;
				\ForEach {child group graph pattern node $P$ of $e_i$}
				{
					$u \leftarrow u \ \cup_{bag} \ \FMain(D, T(P))$\;
				}
				$r \leftarrow r \ \Join \ u$\;
			}
			\ElseIf {$e_i$ is an \texttt{OPTIONAL} node}
			{
				Get the child group graph pattern node $P$ of $e_i$\;
				$o \leftarrow \FMain(D, T(P))$\;
				$r \leftarrow r \ {\tiny \textbf{\textifsym{d|><|}}} \ o$\;
			}
		}
		\nop{\ForEach {child node $e_i$ of $root$}
		{
			\If{$e_i$ is a \texttt{FILTER} node}
			{
			    Get the built-in condition $C$ of $e_i$\;
			    $r \leftarrow \{ \mu \ | \ \mu \in r \wedge {\mu(C)} \}$
			}
		}}
		\Return $r$
	}
\end{algorithm}

Algorithm \ref{alg:bgpe} shows the pseudocode of the BGP-based solution for answering SPARQL query $Q$ based on BE-tree $T(Q)$. The basic idea is to rely on the underlying BGP evaluation engine to evaluate each BGP separately, and combine the results afterwards based on the BE-tree. The return variable $r$, which indicates the result set, is first initialized to be empty (Line 2). Then the child nodes of the BE-tree's root are processed (Lines 4-20).

During the iteration across child nodes of the root, the following cases are considered:
\begin{itemize}
    \item If the current child node is a group graph pattern node, it is recursively evaluated by calling the function on the subtree rooted at it, and the retrieved results are joined with $r$ (Lines 5-9).
    \item If the current child node is a BGP node, it is evaluated by some existing BGP query evaluation technique, and the retrieved results are joined with $r$ (Lines 10-11).
    \item If the current child node is a \texttt{UNION} node, each of its child group graph pattern nodes is recursively evaluated, the results of which are merged by the $\cup_{bag}$ operation. The merged result is finally joined with $r$ (Lines 12-16).
    \item If the current child node is an \texttt{OPTIONAL} node, its child group graph pattern node is recursively evaluated, and the retrieved results are left-outer-joined with $r$ (Lines 17-20).
\end{itemize}

\nop{In the second pass, all the \texttt{FILTER} nodes are handled by filtering $r$ by their built-in conditions (Lines 21-24). \texttt{FILTER} expressions cannot be handled in the first pass because the scope of them is the entire group graph pattern.}

\subsection{BE-Tree Transformations}
\label{sec:transformation}


In the previous subsection, we invoke the BGP-based evaluation procedure (Algorithm \ref{alg:bgpe}) on the BE-tree directly constructed from the query. However, it is possible to improve the efficiency of query evaluation by altering the plan. We achieve this by making certain semantics-preserving \emph{transformations} to the BE-tree.


\subsubsection{Goals}
Our aim is to transform the original BE-tree so that the resulting BE-tree has the following properties:

\begin{itemize}
    \item \emph{Validity:} the resulting BE-tree should maintain the previously defined tree structure and have the same node types. It should be one-to-one mapped to a syntactically valid SPARQL query by the direct construction process introduced in Section \ref{sec:bgp}.
    \item \emph{Efficiency:} the evaluation of the resulting BE-tree should be more efficient than the original BE-tree. In other words, the expected cost of evaluating the resulting BE-tree should be lower.
\end{itemize}

\subsubsection{Semantics-Preserving Transformations}

We set out to transform the BE-tree with the two aforementioned goals in mind.
In order to optimize for query execution efficiency while maintaining correctness, we need to leverage the inherent semantic equivalences regarding the \texttt{UNION} and \texttt{OPTIONAL} operators, formally expressed through the following two theorems.

\begin{theorem}
For any graph pattern $P_1$, $P_2$, $P_3$ and any RDF dataset $D$, we have
    {\small \begin{equation}
        [\![P_1 \ \texttt{AND} \ (P_2  \ \texttt{UNION} \ P_3)]\!]_D = [\![(P_1 \ \texttt{AND} \ P_2) \ \texttt{UNION} \ (P_1 \ \texttt{AND} \ P_3)]\!]_D. \nonumber
    \end{equation}}
\label{thm:union_eq}
\end{theorem}

\begin{proof}
By Definition \ref{def:eval} and the definitions of the operators on bags, we have
    {\small \begin{equation}
    \begin{split}
        & [\![P_1 \ \texttt{AND} \ (P_2  \ \texttt{UNION} \ P_3)]\!]_D \\
        = \ & [\![P_1]\!]_D \Join [\![P_2  \ \texttt{UNION} \ P_3]\!]_D \\
        = \ & [\![P_1]\!]_D \Join ([\![P_2]\!]_D \ \cup_{bag} \ [\![P_3]\!]_D) \\
        = \ & ([\![P_1]\!]_D \Join [\![P_2]\!]_D) \cup_{bag} ([\![P_1]\!]_D \Join [\![P_3]\!]_D) \\
        = \ & [\![P_1 \ \texttt{AND} \ P_2]\!]_D \cup_{bag} [\![P_1 \ \texttt{AND} \ P_3]\!]_D \\
        = \ & [\![(P_1 \ \texttt{AND} \ P_2) \ \texttt{UNION} \ (P_1 \ \texttt{AND} \ P_3)]\!]_D. \nonumber
    \end{split}
    \end{equation}}
\end{proof}

Note that Theorem \ref{thm:union_eq} is also trivially extendable to \texttt{UNION} nodes with more than two child nodes.

\begin{theorem}
For any graph pattern $P_1$, $P_2$ and any RDF dataset $D$, we have
    {\small \begin{equation}
        [\![P_1 \ \texttt{OPTIONAL} \ P_2]\!]_D = [\![P_1 \ \texttt{OPTIONAL} \ (P_1 \ \texttt{AND} \ P_2)]\!]_D. \nonumber
    \end{equation}}
\label{thm:optional_eq}
\end{theorem}
These two equivalences correspond to two semantics-preserving transformations on the BE-tree: that of \emph{merging} a node with the child nodes of its sibling \texttt{UNION} node, and that of \emph{injecting} a node into the child node of its sibling \texttt{OPTIONAL} node. We define these transformations as follows.

\begin{proof}
Similarly, we have
    {\small \begin{equation}
    \begin{split}
        & [\![P_1 \ \texttt{OPTIONAL} \ (P_1 \ \texttt{AND} \ P_2)]\!]_D \\
        = \ & ([\![P_1]\!]_D \Join [\![P_1 \ \texttt{AND} \ P_2]\!]_D) \cup_{bag} ([\![P_1]\!]_D \setminus [\![P_1 \ \texttt{AND} \ P_2]\!]_D) \\
        = \ & ([\![P_1]\!]_D \Join ([\![P_1]\!]_D \Join [\![P_2]\!]_D)) \\ 
        & \cup_{bag} ([\![P_1]\!]_D \setminus ([\![P_1]\!]_D \Join [\![P_2]\!]_D)) \\
        = \ & ([\![P_1]\!]_D \Join [\![P_2]\!]_D) \cup_{bag} ([\![P_1]\!]_D \setminus [\![P_2]\!]_D) \\
        = \ & [\![P_1 \ \texttt{OPTIONAL} \ P_2]\!]_D. \nonumber
    \end{split}
    \end{equation}}
\end{proof}


\begin{definition}[Merge transformation] 
\label{def:merge}
A \underline{merge} transformation is the action performed on a node, which represents the graph pattern $P_1$, and one of its sibling \texttt{UNION} nodes, the child nodes of which represents the group graph patterns $P_2, P_3, \cdots, P_n$, when both of the following conditions are met:
\begin{enumerate}
    \item $P_1$ is a BGP node;
    \item At least one of the group graph patterns in $P_2, P_3, \cdots, P_n$ is the parent node of a BGP node that is coalescable with $P_1$.
\end{enumerate}

The action consists of the following steps:
\begin{enumerate}
    \item Insert $P_1$ 
    as the leftmost child node of $P_2, P_3, \cdots, P_n$;
    \item Coalesce $P_1$ 
    with the other BGP child nodes if possible, until all the BGP nodes are maximal;
    \item Remove $P_1$ from its original position.
\end{enumerate}
\end{definition}

\begin{definition}[Inject transformation] 
\label{def:inject}
An \underline{inject} transformation is the action performed on a node, which represents the graph pattern $P_1$, and one of its sibling \texttt{OPTIONAL} nodes to its right, the child node of which represents the group graph pattern $P_2$, when both of the following conditions are met:
\begin{enumerate}
    \item $P_1$ is a BGP node;
    \item $P_2$ is the parent node of a BGP node that is coalescable with $P_1$.
\end{enumerate}

The action consists of the following steps:
\begin{enumerate}
    \item Insert $P_1$ 
    as the leftmost child node of $P_2$;
    \item Coalesce $P_1$ 
    with the other BGP child nodes if possible, until all the BGP child nodes are maximal.
\end{enumerate}
\end{definition}


\nop{
\begin{definition}[Sub-BGP] 
\label{def:sub-bgp}
A BGP $b'$ formed from one or more triple patterns of the BGP $b$ is called a \underline{sub-BGP} of $b$. Note that each of these triple patterns must be coalescable with one of the others by Definition \ref{def:bgp}.
\end{definition}
}



\begin{figure}[ht]
	\centering
	\includegraphics[scale=0.5]{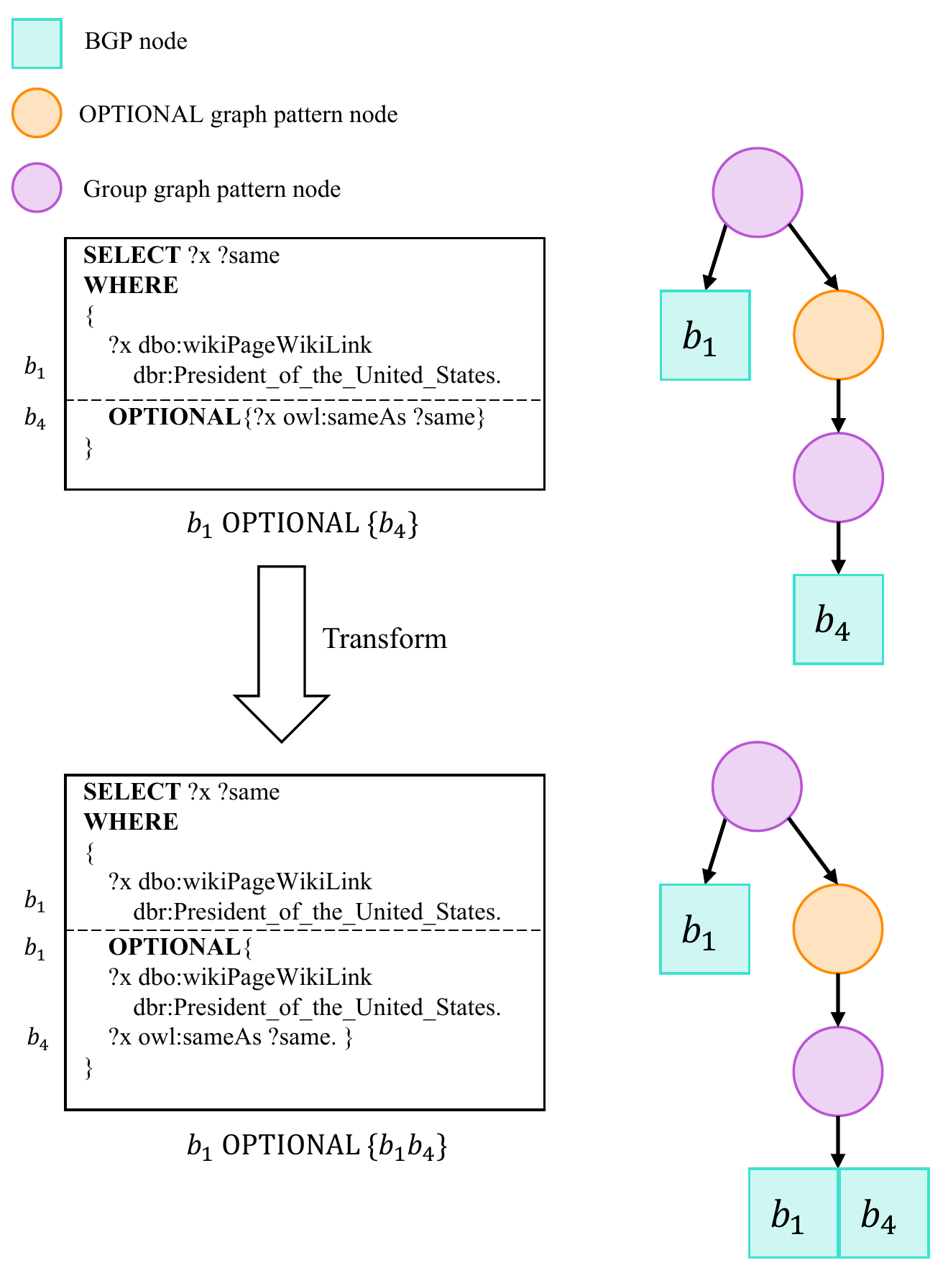}
	\vspace{-0.1in}
	\caption{Favorable \emph{Inject} Transformation}
	\label{fig:example_optional}
	\vspace{-0.1in}
\end{figure}

\begin{figure}[ht]
	\centering
	\includegraphics[scale=0.5]{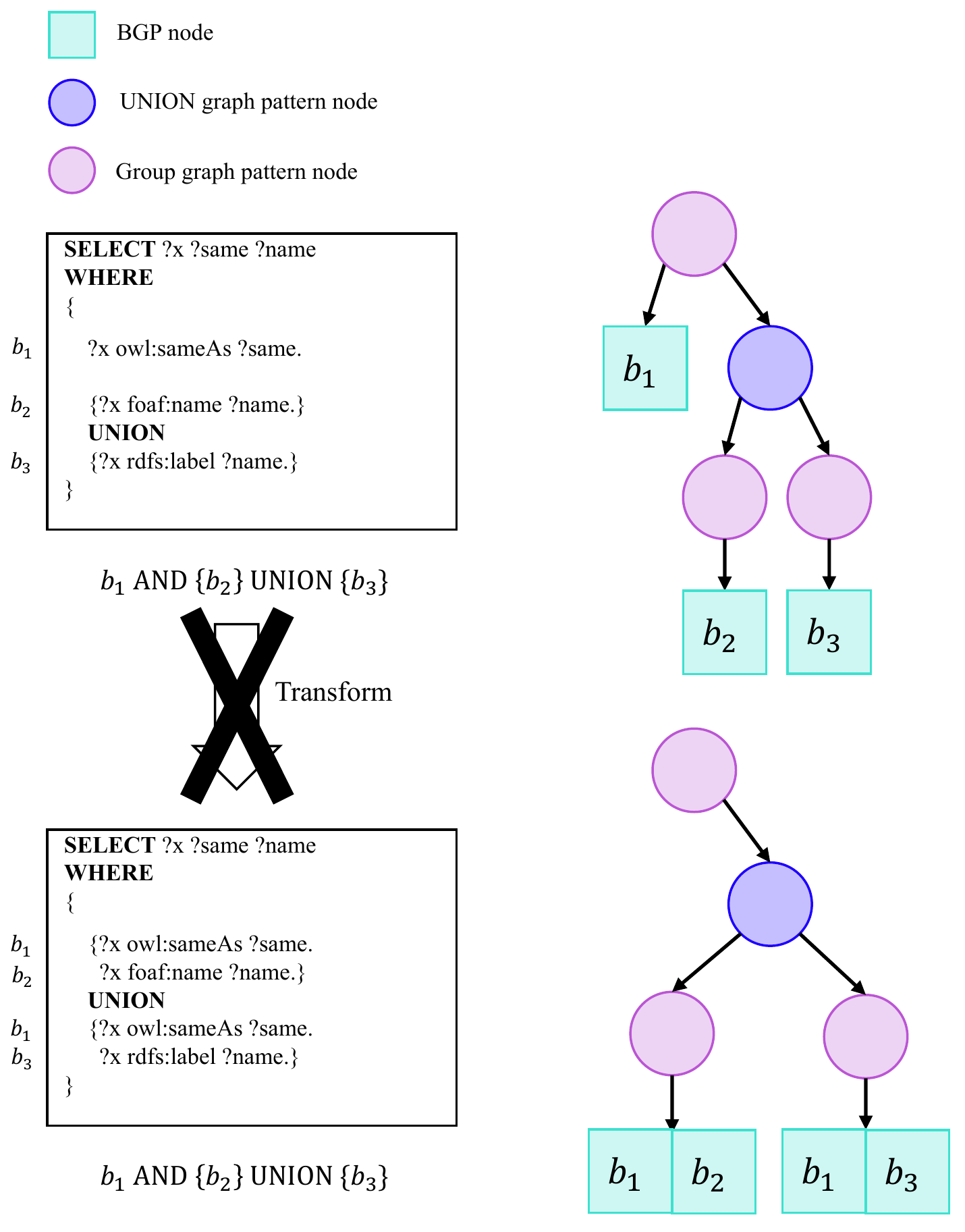}
	\vspace{-0.1in}
	\caption{Unfavorable \emph{Merge} Transformation}
	\label{fig:example_union}
	\vspace{-0.1in}
\end{figure}

Figures \ref{fig:example_optional} and \ref{fig:example_union} are examples of these two types of transformations in action. The graph database targeted by the queries in these figures is DBpedia, which is an encyclopedic open-domain knowledge graph containing information about a vast number of real-world entities. Via these examples, we give a qualitative overview of the effects of these transformations on the plan's efficiency.

In Figure \ref{fig:example_optional}, $b_4$ is a grandchild BGP node of the \texttt{OPTIONAL} node to the right of $b_1$, and $b_1$ and $b_4$ are coalescable. Therefore, the available \emph{inject} transformation will coalesce $b_4$ with $b_1$, which can help improve efficiency. According to the original BE-tree, $b_4$ is directly evaluated, and the results are left-outer-joined with those of $b_1$. Since a large number of entities have the \texttt{?sameAs} relation, which denotes the equivalence between references to the same real-world object, $b_4$ has many matches, causing both its evaluation and the left-outer-join to be costly. However, presidents of the United States is a minority of the entities, making $b_1$ highly selective, which the \emph{inject} transformation takes advantage of. After the \emph{inject}, we can rely on the underlying evaluation engine to efficiently evaluate the coalesced $b_1 b_4$ by choosing a join order that evaluates the much more selective $b_1$ first. The left-outer-join is also rendered less expensive due to the decrease in the number of results of $b_1 b_4$ compared with $b_4$.

This example also helps explain the reason why certain conditions need to be met in Definitions \ref{def:merge} and \ref{def:inject}. It is observable that only by coalescing BGPs is it possible to accelerate BGP evaluation. If no coalescing happens, the repetitive evaluation of the merged or injected BGP will instead incur extra overhead.

However, not all available transformations can help improve efficiency. Figure \ref{fig:example_union} shows an available \emph{merge} transformation on an example \texttt{UNION} query, which merges the BGP $b_1$ with its sibling \texttt{UNION} node. Since $b_1$ has low selectivity, merging it does not accelerate BGP evaluation or reduce the number of intermediate results, and even incurs extra overhead because it now has to be evaluated twice.



%% file: method.tex
\section{Cost-Driven Plan Selection}\label{sec:plan_selection}

In the previous section, we have established that there are differences in terms of efficiency among different semantics-preserving BE-tree transformations. This is the classical cost-based query plan selection problem. In this section, we introduce the cost model for evaluating a \emph{merge} or \emph{inject} transformation, and the algorithm based on it that decides the transformations to be performed given an original BE-tree.

Our cost model handles the BE-tree, and thus operates on a higher level than BGP evaluation. Nevertheless, our cost model still relies on estimations of the evaluation costs and result sizes of BGPs, which are obtainable as long as the workings of the underlying BGP evaluation engine are transparent. In addition, such estimations can often be directly obtained from the plan generation module of the underlying BGP evaluation engines \cite{DBLP:journals/pvldb/MhedhbiS19}. For completeness, we briefly introduce the BGP cost model of gStore \cite{DBLP:journals/pvldb/ZouMCOZ11} and Jena \cite{Wilkinson:Jena2}, the two systems on which we implement our approach for  experimentation in Section \ref{subsubsec:cost_bgp}.

\subsection{Cost Models}

\subsubsection{Cost Model for SPARQL-UO}\label{subsubsec:cost_uo}

The basic idea of our cost model is the insight drawn from the previous examples (Figures \ref{fig:example_optional} and \ref{fig:example_union}): SPARQL-UO query execution cost is made up of two main components:  the cost of evaluating BGPs and the cost of combining partial results through \texttt{UNION}, \texttt{OPTIONAL}, or implicit \texttt{AND} operations. We are primarily concerned with the cost difference caused by a transformation, which we call \emph{$\Delta$-cost}. A transformation is expected to improve efficiency only when its $\Delta$-cost is negative, indicating a decrease in cost; naturally we are looking for the transformation with the most negative $\Delta$-cost.




\begin{figure}[ht]
	\centering
	\includegraphics[scale=0.45]{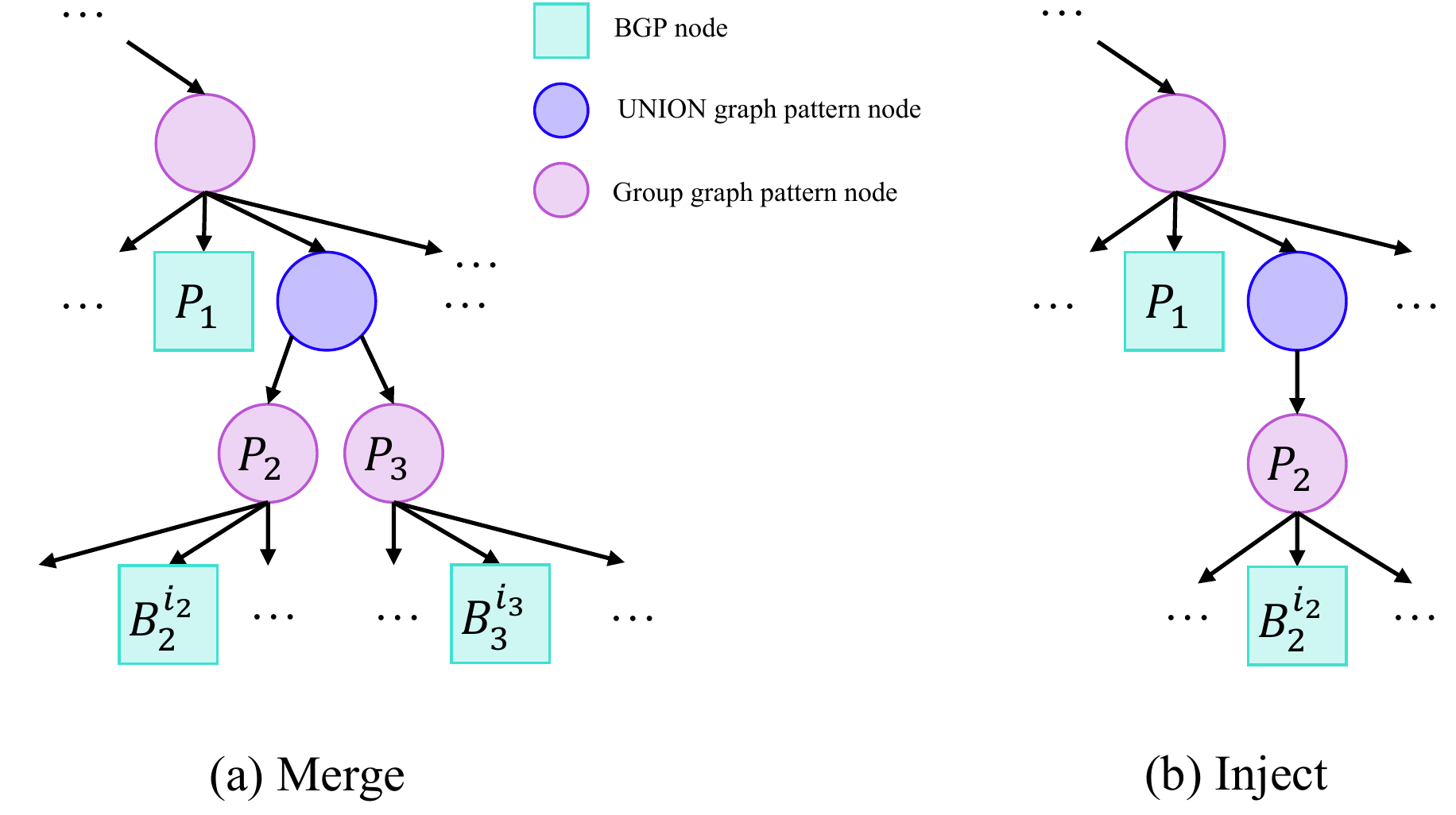}
	\vspace{-0.1in}
	\caption{Estimating the $\Delta$-cost for the \emph{merge} transformation}
	\label{fig:union_cost}
	\vspace{-0.1in}
\end{figure}

In the following, we discuss how to estimate the $\Delta$-cost. Consider a part of BE-tree shown in Figure \ref{fig:union_cost} that contains a \texttt{UNION} node with two  children group graph pattern nodes, $P_2$ and $P_3$ that have  BGP child nodes, $B_2^{i_2}$ and $B_3^{i_3}$, respectively. Assume that  $B_2^{i_2}$ and $B_3^{i_3}$ are coalescable with $P_1$. According to condition (2) in Definition \ref{def:merge}, at most one of $P_2$ and $P_3$ lacks a coalescable BGP child node, which is represented by $B_2^{i_2}$ or $B_3^{i_3}$ as an empty node.


A \emph{merge} transformation only affects a BGP node ($P_1$) and the BGP child nodes of its sibling \texttt{UNION} node ($B_2^{i_2}$ and $B_3^{i_3}$). After the transformation, the constituent triple patterns of these BGP nodes may change, but the occurrence of these nodes are maintained (for empty BGP nodes resulting from transformations are retained). Also, since the transformation preserves the query semantics, the evaluation results of $P_1$'s parent node will not change. Therefore, the cost difference caused by \emph{merge} is local to these nodes and their siblings, as shown in Figure \ref{fig:union_cost}(a). The local cost of the BE-tree before the \emph{merge} transformation $t_m$ can then be estimated as follows, where $l(\cdot)$ and $r(\cdot)$ denote all the sibling nodes to the left and right of the node $\cdot$, respectively:

{\small
\begin{align}
    cost(t_m) &= cost(t_m, BGP) + cost(t_m, algebra) \label{eq:cost_overall} \\
    cost(t_m, BGP) &= cost({P_1}) + cost({B_2^{i_2}}) + cost({B_3^{i_3}}) \label{eq:union_cost_BGP} \\
    cost(t_m, algebra) &= f_{AND}(|res({P_1})|, |res(l({P_1}))|, |res(r({P_1}))|) \nonumber \\
    &+ f_{AND}(|res({B_2^{i_2}})|, |res(l({B_2^{i_2}}))|, |res(r({B_2^{i_2}}))|) \nonumber \\
    &+ f_{AND}(|res({B_3^{i_3}})|, |res(l({B_3^{i_3}}))|, |res(r({B_3^{i_3}}))|) \nonumber \\
    &+ f_{UNION}(|res(P_2)|, |res(P_3)|) \label{eq:union_cost_algebra}
\end{align}}


$cost(t_m, BGP)$ can be directly obtained according to the BGP evaluation engine. $cost(t_m, algebra)$ is due to the possible change in the result sizes of the affected BGP nodes. In the case of \emph{merge}, $cost(t_m, algebra)$ consists of the cost of performing implicit \texttt{AND} between the affected BGP nodes and their left and right siblings, and of performing \texttt{UNION} on $P_2$ and $P_3$. The costs of algebraic operations are functions on the result sizes of their operands ($f_{AND}$ and $f_{UNION}$ in Equation \ref{eq:union_cost_algebra}). These functions may differ based on different implementations of these algebraic operations. In our experiments, to fit the system we choose to build our implementation upon, $f_{AND}$ is set to be the product of its arguments, and $f_{UNION}$ is set to be the sum of its arguments.

Note that we need to also estimate the result sizes of some nodes for $\Delta$-cost estimation. A BGP node's result size can be estimated by invoking or simulating an estimation module of the underlying BGP evaluation engine. The result sizes of other types of nodes need to be estimated based on an assumed distribution of data. In our experiments, we simply estimate the result size of any join (including \texttt{AND} and \texttt{OPTIONAL}) to be the product of the result sizes of the joined graph patterns, and the result size of \texttt{UNION} to be the sum of the result sizes of the \texttt{UNION}'ed graph patterns.

Suppose after the \emph{merge} transformation, the affected nodes are turned into ${P_1}^\prime$, ${B_2^{i_2}}^{\prime}$ and ${B_3^{i_3}}^\prime$. To estimate the local cost after $t_m$ (denoted as $cost(t_{m}^\prime)$), we simply replace $P_1$, $B_2^{i_2}$ and $B_3^{i_3}$ in Equations \ref{eq:union_cost_BGP} and \ref{eq:union_cost_algebra} by ${P_1}'$, ${B_2^{i_2}}'$ and ${B_3^{i_3}}'$.
Consequently, the $\Delta$-cost of \emph{merge} can be estimated as follows:

{\small
\begin{align}
\Delta cost(t_m) = cost(t_{m}^\prime) - cost(t_m)
\end{align}}


\nop{
\begin{figure}[ht]
	\centering
	\includegraphics[scale=0.65]{new_figure/optional_cost.pdf}
	\vspace{-0.1in}
	\caption{Estimating the $\Delta$cost for the \emph{inject} transformation}
	\label{fig:optional_cost}
	\vspace{-0.3in}
\end{figure}
}
The case is similar for the \emph{inject} operation (Figure \ref{fig:union_cost}(b)). The local cost of the BE-tree before the \emph{inject} transformation $t_i$ can then be estimated as follows:

{\small
\begin{align}
      cost(t_i) &= cost(t_i, BGP) + cost(t_i, algebra) \label{eq:optional_cost_overall} \\
    cost(t_i, BGP) &= cost({P_1}) + cost({B_2^{i_2}}) \label{eq:optional_cost_BGP} \\
    cost(t_i, algebra) &= f_{AND}(|res({P_1})|, |res(l({P_1}))|, |res(r({P_1}))|) \nonumber \\
    &+ f_{AND}(|res({B_2^{i_2}})|, |res(l({B_2^{i_2}}))|, |res(r({B_2^{i_2}}))|) \nonumber \\
    &+ f_{OPTIONAL}(|res({P_1})|, |res(P_2)|) \label{eq:optional_cost_algebra}
\end{align}}

Suppose after the \emph{inject} transformation, the affected nodes are turned into ${P_1}^\prime$ and ${B_2^{i_2}}^\prime$. To estimate the local cost after $t_i$, we simply replace $P_1$ and $B_2^{i_2}$ in Equations \ref{eq:optional_cost_BGP} and \ref{eq:optional_cost_algebra} by ${P_1}^\prime$ and ${B_2^{i_2}}^\prime$.
The $\Delta$-cost of \emph{inject} is then computed as follows:

{\small
\begin{align}
    \Delta cost(t_i) = cost(t_{i}^\prime) - cost(t_i)
\end{align}}

\subsubsection{Cost Model for BGP}\label{subsubsec:cost_bgp} Although the underlying BGP cost model is transparent to our SPARQL-UO cost model (see Equations \ref{eq:union_cost_BGP} and \ref{eq:optional_cost_BGP}), for the completeness of exposition, we briefly introduce the BGP cost models employed by gStore and Jena. The evaluation of BGPs consists of joins. Thus the cost of a BGP plan $T$ is the sum of the costs of each executed join operation $j$:
{\small
\begin{align}
    cost(T) = \sum_{j \in T} cost(j) \nonumber
\end{align}}
BGP evaluation in gStore uses the worst-case optimal (WCO) join, which  is concerned with all the edges labeled with the required predicate that links existing query vertices and the newly extended vertex. For each result tuple on the existing vertices, all such edges need to be scanned at least once to check whether this tuple can be extended to match the newly extended vertex. Suppose the set of existing vertices is $\{v_1, \cdots, v_{k-1}\}$, and the newly extended vertex is $v_k$. The cost of a WCO join can then be estimated as follows:
{\small
\begin{align}
    & cost(WCOJoin(\{v_1, \cdots, v_{k-1}\}, v_k)) \nonumber \\
    & = card(\{v_1, \cdots, v_{k-1}\}) \times \min_{i \in [1, k-1]} average\_size(v_i, p) \nonumber
\end{align}}
where $card(\{v_1, \cdots, v_{k-1}\})$ indicates the estimated \emph{cardinality} -- the estimated number of result tuples on the query vertex set $\{v_1, \cdots, v_{k-1}\}$; and $average\_size(v_i, p)$ indicates the average number of edges (\emph{i.e.,} triples) with $p$ as predicate and $v_i$ as subject or object, depending on the direction of the edge between $v_i$ and $v_k$ in the query.

On the other hand, BGP evaluation in Jena uses the binary join, which is conceptually akin to a hash-join in relational databases. It first hashes the result tuples of the BGP with a smaller result size on the common vertices. Then, for each result tuple of the other BGP, the hash index is probed to find compatible matches that can be combined. Suppose the two BGPs to be combined have query vertex sets $V_1$ and $V_2$, respectively. The cost of a binary join can then be estimated as follows:
{\small
\begin{align}
    & cost(BinaryJoin(V_1, V_2)) \\
    & = 2 \times \min (card(V_1), card(V_2)) + \max (card(V_1), card(V_2)) \nonumber
\end{align}}
where the first part of the sum indicates the cost of building the hash index, and the second part indicates the cost of probing it.

The above cost estimation formulas rely on the cardinality estimation of query vertex sets. Cardinality estimation starts from single triple patterns, whose query vertex set’s exact cardinality can be obtained reading the pre-built indexes of the RDF store using the constants as key. Each time that a new query vertex is added to the set, we sample the candidate result set, and collate how many result tuples can be generated from the sample by extending to the new query vertex. The estimated cardinality is updated by scaling up based on the previous estimation in proportion to the ratio between the number of extended result tuples and the sample size:
{\small
\begin{align}
    card(V_k) = \max (\frac{\#extend}{\#sample} \times card(V_{k-1}), 1) \nonumber
\end{align}}

Note that more sophisticated cardinality estimation approaches and BGP cost models (such as \cite{10.1145/3178876.3186003}) are orthogonal to our contribution in this paper. Experimental results show that our approach optimize SPARQL-UO query processing significantly by considering the simple but effective BGP cost models and cardinality estimation methods shown above. 



\subsection{Cost-Driven Transformation}

In this subsection, we discuss BE-tree transformation algorithms that leverage the cost model discussed above to decide on transformations for obtaining the most efficient query plan for execution.


\subsubsection{Transforming a Single Tree Level}


We first concentrate on the simpler case where  only  transformations at a single level are considered.



When a BGP node only has a sibling \texttt{UNION} or \texttt{OPTIONAL} node, deciding the transformations is already covered by the cost model introduced in Section \ref{subsubsec:cost_uo}. However, in reality, multiple sibling \texttt{UNION} or \texttt{OPTIONAL} nodes may be viable for transformation. Note that according to Theorems \ref{thm:union_eq} and \ref{thm:optional_eq}, a merged BGP is removed from its original position, while an injected BGP maintains its original occurrence. This means that a BGP can only be merged with one of its sibling \texttt{UNION} nodes, but can be injected into multiple sibling \texttt{OPTIONAL} nodes. Therefore, in order to decide on a \emph{merge} transformation, we need to look holistically at all the \texttt{UNION} nodes at that level, and choose the transformation that incurs the lowest $\Delta$cost. On the other hand, \emph{inject} transformations are mutually independent, so we scan over each \texttt{OPTIONAL} node to the right of the BGP node, and decide individually which ones are worthy of a transformation based on the $\Delta-$cost.

The transformation decision at a single level of the BE-tree is given in Algorithm \ref{alg:single_level}. The transformation happens at the level of the children of the input group graph pattern node $P$. Note that the \emph{merge} transformation of a BGP node can only be determined and performed after iterating over all its sibling UNION nodes (Line 14), while the \emph{inject} operation is decided individually on each sibling \texttt{OPTIONAL} node (Line 16). The subroutines that compute the $\Delta$-cost of each possible transformation is presented in Algorithm \ref{alg:subroutines}.

\begin{algorithm}[ht]\small
	\caption{\small Single-level BE-tree transformation}
	\label{alg:single_level}
	\SetKwInOut{KwIn}{Input}
	\KwIn{RDF dataset $D$, BE-tree $T(Q)$, a group graph pattern node $P$}
	
	\SetKwFunction{FMain}{SingleLevelTransform}
	\SetKwFunction{FMerge}{DecideMerge}
	\SetKwFunction{FInject}{DecideInject}
	\SetKwProg{Fn}{Function}{:}{}
	\SetKw{Continue}{continue}
	\Fn{\FMain{$D, T(Q), P$}}
	{
		\ForEach{$P_1$ in the child nodes of $P$}
		{
		    \If {$P_1$ is a BGP node}
		    {
		        minUnionCost $\leftarrow 0$\;
		        targetUNION $\leftarrow$ empty node\;
		        \ForEach{\texttt{UNION} node $u$ in the child nodes of $P$}
		        {
		            minUnionCostCur $\leftarrow$ \FMerge($P_1$, $u$)\;
		            \If{minUnionCostCur $<$ minUnionCost}
		            {
		                minUnionCost $\leftarrow$ minUnionCostCur\;
		                targetUNION $\leftarrow$ $U$\;
		            }
		        }
		        \If{minUnionCost $< 0$}
		        { Perform \emph{merge} on subBGPglobal and targetUNION }
		        \ForEach{\texttt{OPTIONAL} node $o$ to the right of $P_1$ in the child nodes of $P$}
		        { \FInject($P_1$, $O$)\; }
		    }
		}
	}
\end{algorithm}
	
\begin{algorithm}[ht]\small
	\caption{\small Subroutines for BE-tree transformation}
	\label{alg:subroutines}
	\SetKwInOut{KwIn}{Input}
	
	\SetKwFunction{FMain}{SingleLevelTransform}
	\SetKwFunction{FMerge}{DecideMerge}
	\SetKwFunction{FInject}{DecideInject}
	\SetKwProg{Fn}{Function}{:}{}
	\SetKw{Continue}{continue}
	\Fn{\FMerge{$P_1, U$}}
	{
	    \nop{\If{maximal common sub-BGP of $U$ is coalescable with $P_1$}
	    {
	        Add the maximal common sub-BGP to $P_1$\;
	        Subtract the maximal common sub-BGP from $U$\;
	    }}
	    \If{constraints are violated}{\Return 0\;}
	    originalCost $\leftarrow$ local cost (Equations \ref{eq:cost_overall}, \ref{eq:union_cost_BGP} and \ref{eq:union_cost_algebra})\;
	    minUnionCostCur $\leftarrow$ 0\;
	        \ForEach{child group graph pattern node $P_j$ of $U$}
	        {
	            $BSet_j \leftarrow \{B_j^i | B_j^i$ is a BGP child node of $P_j$ coalescable with $P_1 \}$\;
	            \If{$BSet_j = \emptyset$}
	            { Add an empty BGP node to $BSet_j$\; }
	        }
	        \ForEach{tuple $(B_2^{i_2}, B_3^{i_3}, \cdots)$ drawn from $BSet_2, BSet_3, \cdots$}
	        {
	            Perform \emph{merge} on $P_1$ and $U$\;
	            transformedCost $\leftarrow$ local cost (Equations \ref{eq:cost_overall}, \ref{eq:union_cost_BGP} and \ref{eq:union_cost_algebra})\;
	            $\Delta$cost $\leftarrow$ transformedCost - originalCost\;
	            \If{$\Delta$cost $<$ minUnionCostCur}
	            {
	                minUnionCostCur $\leftarrow$ $\Delta$cost\;
	            }
	            Undo \emph{merge}\;
	        }
	    \Return minUnionCostCur\;
	}
	\Fn{\FInject{$P_1, O$}}
	{
	    \If{constraints are violated}{\Return\;}
	    originalCost $\leftarrow$ local cost (Equations \ref{eq:cost_overall}, \ref{eq:optional_cost_BGP} and \ref{eq:optional_cost_algebra})\;
	        \ForEach{BGP child node $B_2^i$ of $O$'s child group graph pattern node $P_2$ coalescable with $P_1$}
	        {
	            Perform \emph{inject} on $P_1$ and $U$ (coalescing $sub$ with $B_2^i$)\;
	            transformedCost $\leftarrow$ local cost (Equations \ref{eq:cost_overall}, \ref{eq:optional_cost_BGP} and \ref{eq:optional_cost_algebra})\;
	            $\Delta$cost $\leftarrow$ transformedCost - originalCost\;
	            \If{$\Delta$cost $\geq$ 0}
	            {
	              Undo \emph{inject}\;  
	            }
	        }
	}
\end{algorithm}

\subsubsection{Handling Multiple Levels}


Handling the entire BE-tree, which often consists of multiple levels is particularly challenging because of the possible interdependence between transformations across different levels. For example, if we consider transforming the group graph pattern $\{P_1 \ \texttt{OPTIONAL} \ \{P_2 \ \texttt{OPTIONAL} \ P_3\}\}$ ($P_1$, $P_2$ and $P_3$ are all coalescable BGPs), 
there are $2^3$ possible transformations involving whether $P_1$ is injected into $P_2$, whether $P_2$ is injected into $P_3$, and whether $P_1$ is injected into $P_3$. This results in a plan space that 
is exponential in terms of the depth of the BE-tree. In fact, we conjecture that finding the optimal transformation on the entire BE-tree is an NP-hard combinatorial optimization problem.

In order to balance the time complexity and the efficiency of the transformed tree, we propose a greedy strategy to decide on the transformations on the entire BE-tree (Algorithm \ref{alg:multi_level}). Specifically, we traverse the BE-tree in a post-order depth-first fashion. Only when all the child nodes of a group graph pattern node have been traversed (Lines 4-12) do we consider the possible transformations on the level of its children (Line 13, which invokes Algorithm \ref{alg:single_level}). In this way, we ensure that all the lower levels have been appropriately transformed before considering transforming the current level, and the entire transformed tree is guaranteed to be more efficient than the original without expensive backtracking.

\begin{algorithm}[ht]\small
	\caption{\small Multi-level BE-tree transformation}
	\label{alg:multi_level}
	\SetKwInOut{KwIn}{Input}
	\KwIn{RDF dataset $D$, BE-tree $T(Q)$}
	
	\SetKwFunction{FMain}{MultiLevelTransform}
	\SetKwFunction{FSingle}{PostOrderTraverse}
	\SetKwFunction{FPrev}{SingleLevelTransform}
	\SetKwProg{Fn}{Function}{:}{}
	\SetKw{Continue}{continue}
	\Fn{\FMain{$D, T(Q)$}}
	{
		\FSingle($D, T(Q), Q$)\;
	}
	\Fn{\FSingle{$D, T(Q), P$}}
	{
		\ForEach{$P_1$ in the child nodes of $P$}
		{
		    \If{$P_1$ is a group graph pattern node}
		    { \FSingle($D, T(Q), P_1$)\; }
		    \ElseIf{$P_1$ is a \texttt{UNION} node}
		    {
		        \ForEach{child group graph pattern node $P_i$ of $P_1$}
		        { \FSingle($D, T(Q), P_i$)\; }
		    }
		    \ElseIf{$P_1$ is an \texttt{OPTIONAL} node}
		    {
		        Get the child group graph pattern node $P_2$ of $P_1$\;
		        \FSingle($D, T(Q), P_2$)\;
		    }
		}
		\FPrev($D, T(Q), P$)\;
	}
\end{algorithm}


After applying the transformations, the BE-tree still maintains the tree structure and has the same node types. The semantic correctness of the transformed BE-tree is guaranteed by Theorems \ref{thm:union_eq} and \ref{thm:optional_eq}. Therefore, the evaluation algorithm (Algorithm \ref{alg:bgpe}) can still be invoked to evaluate the transformed BE-tree.


\section{Query-Time Optimization: Candidate Pruning}\label{sec:cand_pruning}

In the previous section, we introduced how to generate different SPARQL-UO query plans by BE-tree transformations and select an effective plan based on the cost estimation prior to execution. In this section, we present \emph{candidate pruning}, a query-time optimization incorporated into Algorithm \ref{alg:bgpe} to enhance efficiency.

The basic idea of candidate pruning is also drawn from Theorems \ref{thm:union_eq} and \ref{thm:optional_eq}. The equivalence between the evaluation results implies that the results of the \texttt{UNION}'ed or \texttt{OPTIONAL}-right group graph patterns are constrained by those of the outer graph pattern regarding the common variables.
\nop{A similar constraint exists for nested group graph patterns, according to the following theorem (the proof is obvious due to Definition \ref{def:eval}):

\begin{theorem}
	For any graph pattern $P_1$ and $P_2$ and any \\ RDF dataset $D$, we have
    \begin{equation}
        [\![P_1 \ \texttt{AND} \ \{P_2\}]\!]_D = [\![\{P_1 \ \texttt{AND} \ P_2\}]\!]_D. \nonumber
    \end{equation}
\label{thm:nested_eq}
\end{theorem}}
Therefore, when a \texttt{UNION}, \texttt{OPTIONAL} or group graph pattern node is encountered during evaluation, we can set the current results on the common variables as \emph{candidate results} when executing the child BGPs of that node. Figure \ref{fig:candidate_pruning} shows the mechanism of candidate pruning for an \texttt{OPTIONAL} query: the results of the variable \texttt{?x} from the already evaluated graph patterns serve as the candidate results of \texttt{?x} for the child BGP of the \texttt{OPTIONAL}-right group graph pattern, pruning redundant matchings of \texttt{?x} that will be materialized if the BGP is evaluated independently.

\begin{figure}
    \centering
    \includegraphics[scale=0.5]{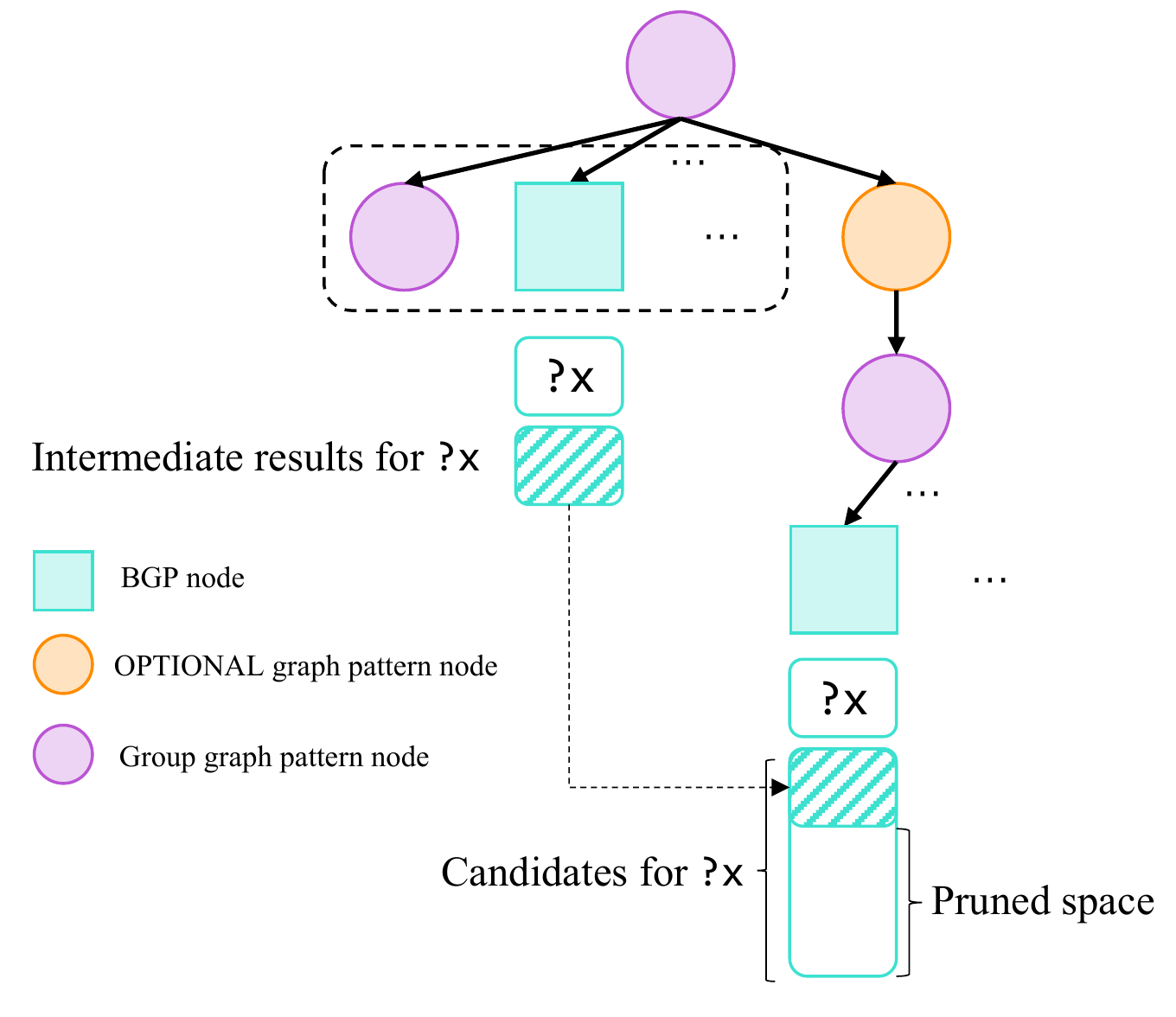}
    \caption{Candidate pruning for \texttt{OPTIONAL}}
    \label{fig:candidate_pruning}
    \vspace{-0.3in}
\end{figure}

The discussion above establishes that candidate pruning preserves semantic correctness. However, to achieve a pruning effect, we need to ensure that the size of the candidate results is
smaller than the size of the actual results of the BGP. A smaller candidate result size also reduces the overhead incurred by scanning them and setting them as candidates. We adopt an adaptive threshold on the candidate result size. The cost model for BGP (Section \ref{subsubsec:cost_bgp}) invoked as part of tree transformation provides an estimate of the actual BGP result size, which we employ as the threshold on candidate result size whenever possible. When no such estimate is available, we set the threshold based on the dataset size. (Please refer to Section \ref{Sec:experiments} for the threshold setting in our experiments.)

To implement candidate pruning, we modify Algorithm \ref{alg:bgpe} as follows (Note that the results can be passed as arguments in the form of pointers to prevent expensive copying):

\begin{itemize}
	\item Add a third argument \emph{cand}, which denotes the candidate results, to the \texttt{BGPBasedEvaluation} function;
	\item Pass the current results \emph{r} as the third argument to \\ \texttt{BGPBasedEvaluation} when processing a \texttt{UNION}, \texttt{OPTIONAL} or group graph pattern node (Lines 7, 9, 15 and 19);
	\item Pass \emph{cand} as the third argument to \texttt{EvaluateBGP} (Line 11). Only when the size of \emph{cand} is smaller than the threshold is it set as the candidate results of the BGP.
\end{itemize}

Tree transformation and candidate pruning, which take effect prior to and during query execution, respectively, are complementary to each other. Prior to execution, high-selectivity BGPs are targeted by \emph{merge} or \emph{inject} transformations, which breaks up graph patterns with large overall results that originally cannot be handled by candidate pruning. Tree transformation also supplies candidate pruning with estimates of the BGP result sizes. On the other hand, while tree transformations are constrained to be performed level-by-level due to the vast plan space, candidate pruning can transmit the pruning effect of small results across levels during execution. For example, when processing a query with the group graph pattern $\{P_1 \ \texttt{OPTIONAL} \ \{P_2 \ \texttt{OPTIONAL} \ P_3\}\}$, $P_1$ cannot be injected into $P_3$ by the greedy transformation strategy even if it is selective, but its results can serve as candidates for $P_3$ via $P_2$. In the special case where there is only a BGP node to the left of the \texttt{UNION} or \texttt{OPTIONAL} node, performing transformations on the BGP is equivalent to candidate pruning. In this case, tree transformation is skipped to evade the additional overhead.


%% file: experiment.tex
\begin{table*}[ht]
	\centering
	\caption{Datasets Statistics}
	\footnotesize
	\begin{tabular}{|c|c|c|c|c|} 
		\hline
		Datasets & triples & entities & predicates & literals \\
		\hline
		LUBM  & 534,355,247 & 86,990,882 & 18 & 44,658,530 \\
		\hline
		DBpedia & 830,030,460 & 96,375,582 & 57,471 & 59,825,935 \\
		\hline
	\end{tabular}
	\label{tab:ds}
	\vspace{-0.1in}
\end{table*}

\begin{table*}[ht]
\begin{minipage}{9cm}	
 \centering
	\caption{Query Statistics on LUBM}
	\footnotesize
	\begin{tabular}{|c|c|c|c|c|c|}
		\hline
		&Query & Type & $Count_{BGP}$ & $Depth$ & $|[\![Q]\!]_D|$ \\ \hline
		\multirow{5}{*}{Group 1}&{q1.1} & {U} & {9} & {2} & {645,666}  \\ \cline{2-6}
		&{q1.2} & {O} & {3} & {2} & {44,653,510}  \\ \cline{2-6} 
		&{q1.3} & {O} & {4} &{4} & {76}  \\ \cline{2-6}  
		&{q1.4} & {O} & {4} &{4} & {5,583}  \\ \cline{2-6} 
		&{q1.5} & {UO} & {6} &{3} & {4,348}  \\ \cline{2-6}
		&{q1.6} & {UO} & {9} &{3} & {37}  \\ 
		\hline \hline
		\multirow{6}{*}{Group 2}&{q2.1} & {O} & {3} &{1} & {4,176,432}  \\ \cline{2-6} 
		&{q2.2} & {O} & {4} &{3} & {8,698}  \\ \cline{2-6} 
		&{q2.3} & {O} & {4} &{3} & {13,124,940}  \\ \cline{2-6}  
		&{q2.4} & {O} & {2} &{3} & {10}  \\ \cline{2-6}  
		&{q2.5} & {O} & {2} &{2} & {10}  \\ \cline{2-6}  
		&{q2.6} & {O} & {2} &{2} & {7} \\ \hline 
	\end{tabular}
	\label{tab:exp_uo_lubm}
 \end{minipage}
 \begin{minipage}{9cm}
 \centering
	\caption{Query Statistics on DBpedia}
	\footnotesize
	\begin{tabular}{|c|c|c|c|c|c|}
		\hline
		& Query & Type & $Count_{BGP}$ & $Depth$ & $|[\![Q]\!]_D|$ \\ \hline
		\multirow{5}{*}{Group 1}&{q1.1} & {U} & {6} &{2} & {153,325}  \\ \cline{2-6}
		&{q1.2} & {UO} & {4} &{3} & {610,434}  \\ \cline{2-6}  
		&{q1.3} & {O} & {5} &{5} & {1,192}  \\ \cline{2-6}  
		&{q1.4} & {UO} & {7} &{5} & {92,041}  \\ \cline{2-6}   
		&{q1.5} & {UO} & {6} &{3} & {3,699,995}  \\ \cline{2-6}
		&{q1.6} & {UO} & {10} &{4} & {176} \\ \hline \hline
		\multirow{6}{*}{Group 2}
		&{q2.1} & {O} & {5} &{3} & {490,876}  \\ \cline{2-6}  
		&{q2.2} & {O} & {2} &{2} & {55,054}  \\ \cline{2-6}
		&{q2.3} & {O} & {2} &{2} & {61,318}  \\ \cline{2-6}   
		&{q2.4} & {O} & {3} &{2} & {4,757} \\ \cline{2-6}  
		&{q2.5} & {O} & {2} &{2} & {5,330}  \\ \cline{2-6} 
		&{q2.6} & {O} & {9} &{2} & {36} \\ 
		\hline
	\end{tabular}
	\label{tab:exp_uo_dbpedia}
 \end{minipage}
\end{table*}

\nop{
\begin{table}[ht]
	\centering
	\caption{Query Statistics on DBpedia}
	\small
	\begin{tabular}{|c|c|c|c|c|c|}
		\hline
		&query & type & $Count_{BGP}(Q)$ & $Depth(Q)$ & $|[\![Q]\!]_D|$ \\ \hline
		\multirow{5}{*}{Group 1}&{q1.1} & {U} & {6} &{2} & {153,325}  \\ \cline{2-6}
		&{q1.2} & {UO} & {4} &{3} & {610,434}  \\ \cline{2-6}  
		&{q1.3} & {O} & {5} &{5} & {1,192}  \\ \cline{2-6}  
		&{q1.4} & {UO} & {7} &{5} & {92,041}  \\ \cline{2-6}   
		&{q1.5} & {UO} & {6} &{3} & {3,699,995}  \\ \cline{2-6}
		&{q1.6} & {UO} & {10} &{4} & {176} \\ \hline \hline
		\multirow{6}{*}{Group 2}
		&{q2.1} & {O} & {5} &{3} & {490,876}  \\ \cline{2-6}  
		&{q2.2} & {O} & {2} &{2} & {55,054}  \\ \cline{2-6}
		&{q2.3} & {O} & {2} &{2} & {61,318}  \\ \cline{2-6}   
		&{q2.4} & {O} & {3} &{2} & {4,757} \\ \cline{2-6}  
		&{q2.5} & {O} & {2} &{2} & {5,330}  \\ \cline{2-6} 
		&{q2.6} & {O} & {9} &{2} & {36} \\ 
		\hline
	\end{tabular}
	\label{tab:exp_uo_dbpedia}
\end{table}}

\nop{
\begin{figure}
	\centering
	\includegraphics[scale=0.5]{exp/Tab5.pdf}
	\vspace{-0.1in}
	\caption{\texttt{UNION} \& \texttt{OPTIONAL} \& \texttt{FILTER}   Performance on LUBM}
	\label{exp:uof-performance:lubm}
	\vspace{-0.1in}
\end{figure}

\begin{figure}
	\centering
	\includegraphics[scale=0.5]{exp/Tab6.pdf}
	\vspace{-0.1in}
	\caption{\texttt{UNION} \& \texttt{OPTIONAL} \& \texttt{FILTER} Performance on DBpedia}
	\label{exp:uof-performance:dbpedia}
	\vspace{-0.1in}
\end{figure}
}
\vspace{-0.2in}
\section{Experiments}\label{Sec:experiments}

To evaluate the effectiveness of our approach, we employ the BGP query engines of Jena and gStore to implement our BGP-based cost-aware SPARQL-UO evaluation strategy. We pull the latest version of Jena as of 27 June, 2022 from their GitHub repository\footnote{\url{https://github.com/apache/jena}.}. All the experiments run on Jena have enabled the statistics-based optimizations. We forked a branch from the main branch of gStore (v0.91) and implement our proposed SPARQL-UO optimizer based on it\footnote{Our implementation is available at \url{https://anonymous.4open.science/r/gStore-UO/}.}. Experiments are conducted on both synthetic (LUBM \cite{lubmurl}) and real (DBpedia\footnote{The DBpedia data dump that we use is V3.9, which is downloadable at \url{http://downloads.dbpedia.org/3.9/en/}. We use the concatenation of all the N-Triples files.}\cite{Dbpediaurl}) RDF datasets, the statistics of which are listed in Table \ref{tab:ds}. Our implementation and all the queries used in our experiments can be found in our anonymous GitHub repository\footnote{\url{https://anonymous.4open.science/r/gStore-UO-opt/}.}. We conduct experiments on a Linux server with an Intel Xeon Gold 6126 CPU @ 2.60GHz CPU and 256GB memory.



\subsection{Verification of Optimizations} \label{subsec:expr_optimizations}

In this section, we verify the effectiveness of the proposed optimizations in Section \ref{sec:transformation} and evaluate the following four approaches:
\begin{enumerate}
    \item The baseline (abbreviated as \texttt{base}), which invokes the BGP-based query evaluation method (Algorithm \ref{alg:bgpe}) on the original BE-tree, analogous to the original SPARQL-UO implementations in Jena and gStore;
    \item Tree transformation (abbreviated as \texttt{TT}), which transforms the original BE-tree by Algorithm \ref{alg:multi_level} and then invokes Algorithm \ref{alg:bgpe} on it;
    \item Candidate pruning (abbreviated as \texttt{CP}), which invokes Algorithm \ref{alg:bgpe} augmented with candidate pruning (Section \ref{sec:cand_pruning}) on the original BE-tree, using a fixed threshold of 1\% of the total number of triples in the database;
    \item The full version that coordinates tree transformation and candidate pruning (abbreviated as \texttt{full}), which transforms the original BE-tree by Algorithm \ref{alg:multi_level}, and then invokes Algorithm \ref{alg:bgpe} augmented by candidate pruning, using an adaptive threshold on the candidate result size.
\end{enumerate}

Since there is no benchmark tailored for SPARQL-UO queries to our knowledge, we construct a mini-benchmark with realistic semantics and varying complexities, containing six queries on LUBM and DBpedia, respectively, denoted as q1.1-1.6 in the following and given in Appendix A of \cite{zouSPARQLUO:22}.
Let $Q$ be the outermost group graph pattern in the query. To measure the complexity of a query, we define two metrics: (1) the BGP count ($Count_{BGP}(Q)$), and (2) the maximum depth of nested group graph patterns ($Depth(Q)$).

$Count_{BGP}(P)$ of a graph pattern $P$ is recursively defined:
\begin{enumerate}
	\item If $P$ is a BGP, $Count_{BGP}(P) = 1$.
	\item If $P = \{P_1\}$, $Count_{BGP}(P) = Count_{BGP}(P_1)$.
	\item If $P = P_1 \ \texttt{AND} \ P_2$ or $P_1 \ \texttt{UNION} \ P_2$ or $P_1 \ \texttt{OPTIONAL} \ P_2$, \\$Count_{BGP}(P) = Count_{BGP}(P_1) + Count_{BGP}(P_2)$.
\end{enumerate}

$Depth(P)$ of a graph pattern $P$ is recursively defined as follows:
\begin{enumerate}
	\item If $P$ is a BGP, $Depth(P) = 0$.
	\item If $P = \{P_1\}$, $Depth(P) = Depth(P_1) + 1$.
	\item If $P = P_1 \ \texttt{AND} \ P_2$ or $P_1 \ \texttt{UNION} \ P_2$ or $P_1 \ \texttt{OPTIONAL} \ P_2$, \\ $Depth(P) = \max (Depth(P_1),Depth(P_2))$.
\end{enumerate}

Suppose $P$ is the outermost group graph pattern of query $Q$, we have $Count_{BGP}(Q) = Count_{BGP}(P), Depth(Q) = Depth(P)$. Group 1 in Tables \ref{tab:exp_uo_lubm} and \ref{tab:exp_uo_dbpedia} summarizes the statistics and the result sizes of the queries used in this subsection. 


\nop{
\begin{figure*}
	\centering
	\begin{subfigure}[b]{0.48\textwidth}
		\centering
		\includegraphics[width=\textwidth]{exp/verify_opt_lubm_gStore.pdf}
		\caption{gStore}
		\label{exp:uo-performance:lubm:gStore}
	\end{subfigure}
	\hfill
	\begin{subfigure}[b]{0.48\textwidth}
		\centering
		\includegraphics[width=\textwidth]{exp/verify_opt_lubm_Jena.pdf}
		\caption{Jena}
		\label{exp:uo-performance:lubm:Jena}
	\end{subfigure}
	\vspace{-0.1in}
	\caption{Verification of optimizations on LUBM}
	\label{exp:uo-performance:lubm}
	\vspace{-0.1in}
\end{figure*}

\begin{figure*}
	\centering
	\begin{subfigure}[b]{0.48\textwidth}
		\centering
		\includegraphics[width=\textwidth]{exp/verify_opt_dbpedia_gStore.pdf}
		\caption{gStore}
		\label{exp:uo-performance:dbpedia:gStore}
	\end{subfigure}
	\hfill
	\begin{subfigure}[b]{0.48\textwidth}
		\centering
		\includegraphics[width=\textwidth]{exp/verify_opt_dbpedia_Jena.pdf}
		\caption{Jena}
		\label{exp:uo-performance:dbpedia:Jena}
	\end{subfigure}
	\vspace{-0.1in}
	\caption{Verification of optimizations on DBpedia}
	\label{exp:uo-performance:dbpedia}
	\vspace{-0.0in}
\end{figure*}
}

\begin{figure*}[!h]
    \centering
	\includegraphics[scale=0.4]{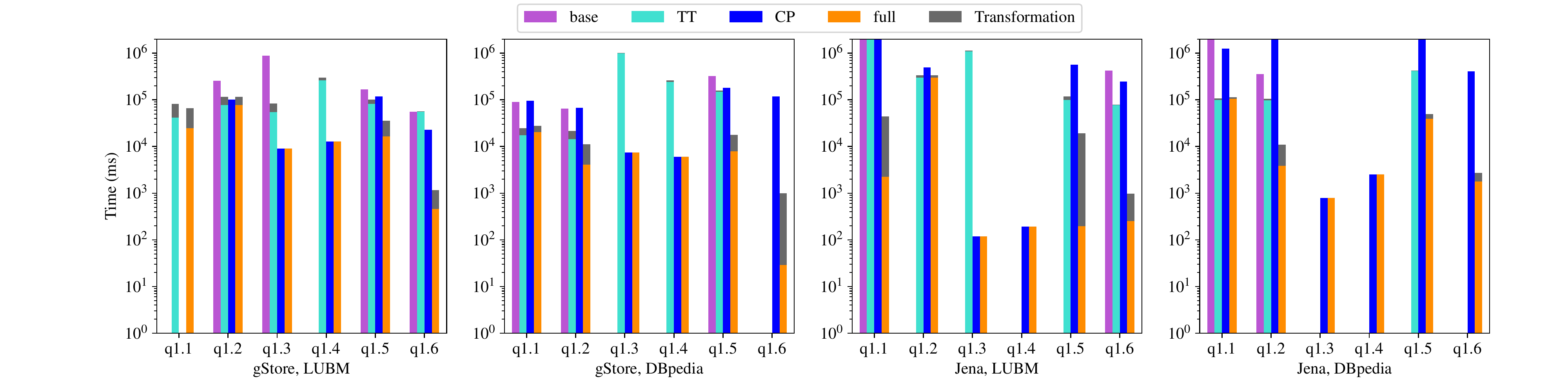}
	\vspace{-0.2in}
	\caption{Verification of optimizations.}
	\label{fig:exp-verify}
	\vspace{-0.1in}
\end{figure*}

\begin{figure*}[!h]
	\centering
	\includegraphics[scale=0.5]{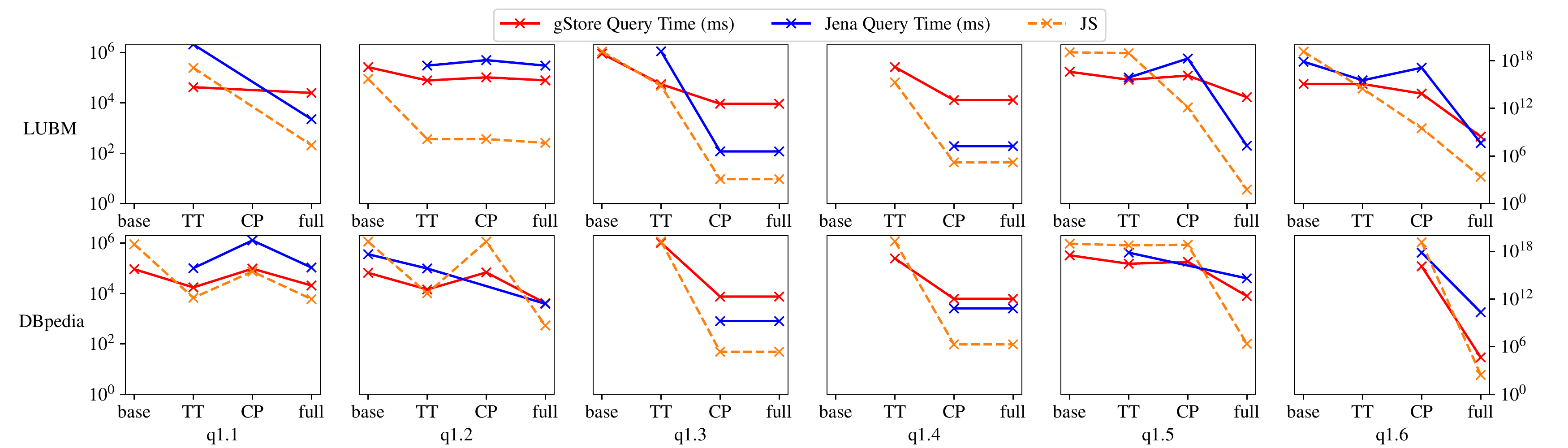}
	\vspace{-0.1in}
	\caption{The execution time and join space of queries.}
	\label{fig:exp-explain}
	\vspace{-0.1in}
\end{figure*}

We measure the performance by the query execution time. We also report the time spent carrying out the tree transformations for \texttt{TT} and \texttt{full}.
The performance of our approaches on LUBM and DBpedia is shown in Figure \ref{fig:exp-verify}.
The absence of a bar indicates an out-of-memory error on the query. A query is considered timed-out if the execution time exceeds $2 \times 10^6$ microseconds.

The trends of the results across gStore and Jena are similar, showing the adaptability of our approach regardless of the underlying BGP execution engine. 
Both of our proposed optimizations are shown to be effective since \texttt{TT}, \texttt{CP} and \texttt{full} perform better than \texttt{base} on all queries. \texttt{TT} and \texttt{CP} can be more advantageous on different queries and datasets than the other. Their optimization effects are cumulative when combined: \texttt{full} performs best all queries and datasets (except on q1.2 on gStore, where \texttt{CP} beats \texttt{full} by a small margin), beating the baseline by at least 2x and up to over an order of magnitude. Our optimized approaches also consume less memory. While \texttt{base} runs out of memory on 13 out of 24 queries, \texttt{full} successfully runs all the queries.

In the following, we try to draw some conclusions about the applicability of our optimizations to different SPARQL-UO queries by analyzing the benchmark queries and the behavior of the optimized approaches on them.

\Paragraph{When \texttt{TT} is effective.} q1.1 on DBpedia (Listing \ref{lst:q1.1_dbpedia}, Appendix \ref{app:queries} of \cite{zouSPARQLUO:22}) is a query on which \texttt{TT} is effective, but \texttt{CP} is not. In this query, two \texttt{UNION} clauses are given first (Lines 2-3), whose child BGPs all have low selectivity. There is no high-selectivity graph pattern before them to enable \texttt{CP}. However, \texttt{TT} can merge the high-selectivity BGP in Lines 5-8 with the \texttt{UNION} clause in Line 3 to accelerate query processing and reduce memory overhead, as evidenced in Figure \ref{fig:exp-verify}. q1.2 on LUBM and q1.2 on DBpedia also belong to this category. (Note that q1.2 on LUBM corresponds to the special case mentioned in Section \ref{sec:cand_pruning}, where there is only a BGP before an \texttt{OPTIONAL} clause, and thus \texttt{TT} and \texttt{CP} have a similar effect.)
\Paragraph{When \texttt{CP} is effective.} q1.3 on LUBM (Listing \ref{lst:q1.3_lubm}, Appendix \ref{app:queries} of \cite{zouSPARQLUO:22}) is a query on which \texttt{CP} is effective, but \texttt{TT} is not. In this query, the BGP in Line 2 has high selectivity, followed by nested \texttt{OPTIONAL}s with low-selectivity child BGPs. \texttt{TT} can inject the BGP into the outermost \texttt{OPTIONAL} but cannot reach the inner \texttt{OPTIONAL}s, thus having limited effect. However, \texttt{CP} can carry the small number of results into the innermost \texttt{OPTIONAL} and set them as candidates to accelerate query processing. q1.3-4 on LUBM and q1.3-4 on DBpedia also belong to this category.
\Paragraph{When \texttt{TT} and \texttt{CP} are jointly effective.} q1.6 on LUBM (Listing \ref{lst:q1.6_lubm}, Appendix \ref{app:queries} of \cite{zouSPARQLUO:22}) is a query on which \texttt{TT} and \texttt{CP} work complementarily, causing \texttt{full} to perform much better than \texttt{TT} and \texttt{CP}. In this query, the BGP in Lines 2-3 has high selectivity, while the BGP in Line 4 has relatively low selectivity. Upon obtaining their considerably large results, \texttt{CP} has limited effect on the following \texttt{UNION} clauses. \texttt{TT}, however, can pick the high-selectivity BGP to merge with the \texttt{UNION} in Line 5. Having executed the graph patterns up to Line 6, \texttt{CP} can accelerate the processing of upcoming \texttt{OPTIONAL}s. q1.1 and q1.5 on LUBM and q1.5 and q1.6 DBpedia also belong to this category.

For a quantitative perspective on the optimization effects, we define the join space of a graph pattern $JS(P)$ as follows:

\begin{enumerate}
	\item If $P$ is a BGP, $JS(P) = |[\![P]\!]_D|$.
	\item If $P = \{P_1\}$, $JS(P) = JS(P_1)$.
	\item If $P = P_1 \ \texttt{AND} \ P_2$ or $P_1 \ \texttt{OPTIONAL} \ P_2$, $JS(P) = JS(P_1) \times JS(P_2)$.
	\item If $P = P_1 \ \texttt{UNION} \ P_2$, $JS(P) = JS(P_1) + JS(P_2)$.
\end{enumerate}

\begin{figure*}
	\centering
	\includegraphics[scale=0.5]{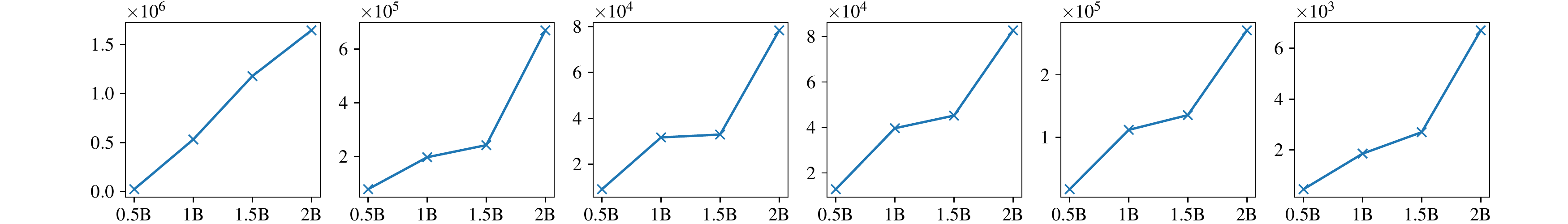}
	\vspace{-0.1in}
	\caption{Query execution time (ms) of \texttt{full} on LUBM datasets of different sizes (``B'' is short for billion).}
	\label{exp:scalability}
	\vspace{-0.1in}
\end{figure*}

The join space of a query estimates the largest intermediate result size that is materialized during the execution of this query. Therefore, it is indicative of both the query's execution time and memory overhead. We plot the execution time of all the queries on gStore and Jena (the y-axis on the left) with their respective join spaces (the y-axis on the right) in Figure \ref{fig:exp-explain}. Across the tested approaches, these three metrics show a similar trend. On all the queries, the join spaces of \texttt{TT} and \texttt{CP} are smaller than those of \texttt{base}, and \texttt{full} has the smallest join space overall, which corroborates the qualitative analysis above.

\subsection{Comparison with State-of-the-Art}

\begin{figure}[ht]
	\centering
	\includegraphics[scale=0.5]{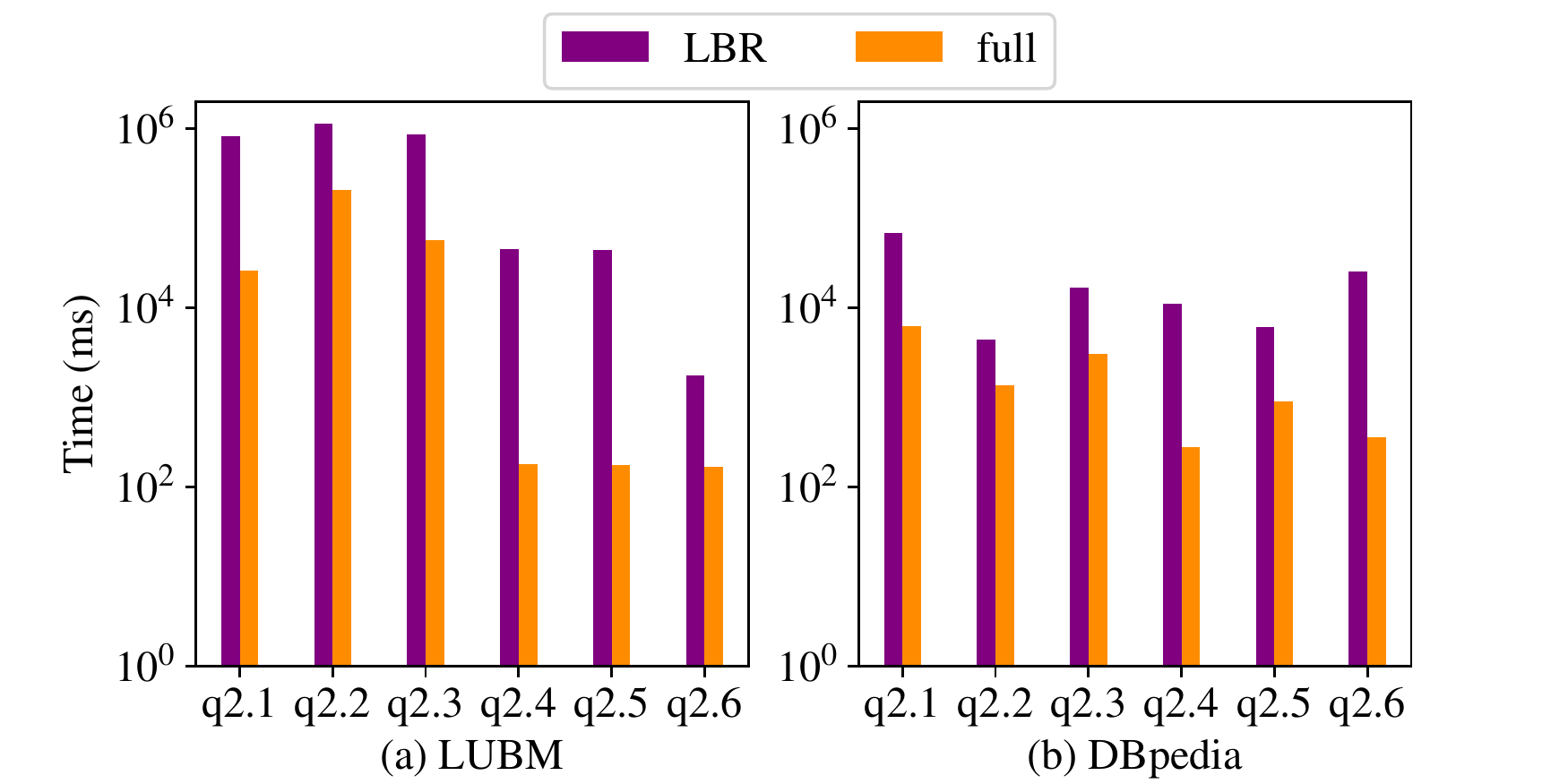}
	\vspace{-0.1in}
	\caption{Comparison with state-of-the-art on LUBM.}
	\label{exp:overall-lubm}
	\vspace{-0.15in}
\end{figure}

\nop{
\begin{figure}[ht]
	\centering
	\includegraphics[scale=0.45]{exp/sota_dbpedia.pdf}
	\vspace{-0.1in}
	\caption{Comparison with state-of-the-art on DBpedia.}
	\label{exp:overall-dbpedia}
	\vspace{-0.1in}
\end{figure}
}

The only work that considers SPARQL with \texttt{OPTIONAL} query optimization is LBR \cite{DBLP:conf/sigmod/Atre15}. Thus, we compare our \texttt{full} approach with LBR. We also implement LBR in C++.
We experiment on the queries provided in LBR \cite{DBLP:conf/sigmod/Atre15} on LUBM and DBpedia, listed as q2.1-2.6, given in Appendix \ref{app:queries} of \cite{zouSPARQLUO:22}. The statistics of these queries are given in the second group in Tables \ref{tab:exp_uo_lubm} and \ref{tab:exp_uo_dbpedia}. q2.1-2.3 are complex with multiple nested group graph patterns, each containing a low-selectivity BGP followed by an \texttt{OPTIONAL} with a single low-selectivity child BGP. Meanwhile, q2.4-2.6 are simple without nested group graph patterns, and their outermost group graph pattern contains a high-selectivity BGP followed by an \texttt{OPTIONAL}.

The total response time of \texttt{full} and LBR are shown in Figures \ref{exp:overall-lubm}.
\texttt{full} is significantly faster than LBR on all queries, and the improvement on q2.4-2.6 is more significant than on q2.1-2.3. This is because candidate pruning can take advantage of the high-selectivity BGPs in q2.4-2.6, while q2.1-2.3 does not contain high-selectivity BGPs. (Note that since all the group graph patterns in q2.1-2.6 contain a BGP followed by an \texttt{OPTIONAL} clause, they correspond to the special case mentioned in Section \ref{sec:cand_pruning} where tree transformation and candidate pruning are equivalent, hence only candidate pruning is performed.) The results show that when candidate pruning takes effect, it is more efficient than LBR's heavy-weight pruning strategies. On q2.1-2.3, \texttt{full} is still faster than LBR since its BGP-based evaluation scheme is more efficient than LBR's separate treatment of triple patterns. 

In summary, our approach outperforms LBR on \texttt{OPTIONAL} queries, despite LBR being optimized for \texttt{OPTIONAL}.


\subsection{Scalability Study}

Lastly, we evaluate how well our approach scales to larger datasets. By setting the scaling factor of LUBM, i.e., the number of universities, we generate three more LUBM datasets with 1, 1.5 and 2 billion triples, respectively. We run the \texttt{full} approach on q1.1-q1.6 on these datasets and plot how the execution time changes with the dataset size on each query in Figure \ref{exp:scalability}.

These plots are empirical complexity curves of our approach.
Our approach scales almost linearly to the number of triples in the datasets. The growth rate of the query execution time correlates with each query's result sizes: the execution time of queries with larger result sizes grows faster with the dataset size. (The result sizes of q1.3-1.6 on larger LUBM datasets are equal to those shown in Table \ref{tab:exp_uo_lubm}, while those of q1.1-1.2 grow linearly.)

%% file: conclusion.tex
\vspace{-0.0in}
\section{Conclusion}\label{sec:conclusion}
The proliferation of knowledge graph applications has generated increasing RDF data management problems. In this paper, we focus on how to optimize SPARQL queries with \texttt{UNION} and \texttt{OPTIONAL} clauses (SPARQL-UO for short). Making use of existing BGP query evaluation modules in SPARQL engines, we propose a series of cost-driven transformations on the BGP-based evaluation tree (BE-tree). These optimizations can significantly reduce the search space and intermediate result sizes, and thus improve both the time and space efficiency of SPARQL-UO query evaluation. We experimentally validate the effectiveness of our optimizations, and compare the performance of the optimized method with the state-of-the-art on large-scale synthetic and real RDF datasets containing millions of triples. These experiments confirm that our SPARQL-UO query evaluation method is orders of magnitude more efficient than existing work.



\clearpage

%% file: appendix.tex
\appendix
{\noindent\Large\textbf{APPENDIX}}
\section{Queries used in Experiments}\label{app:queries}

\subsection{Queries on LUBM}


\begin{lstlisting}[caption={Prefixes of LUBM Queries}, frame=single, breaklines=true, basicstyle=\scriptsize, morekeywords={PREFIX, SELECT, WHERE, UNION, OPTIONAL}]
PREFIX ub: <http://swat.cse.lehigh.edu/onto/univ-bench.owl#>
PREFIX rdf: <http://www.w3.org/1999/02/22-rdf-syntax-ns#>
\end{lstlisting}

\begin{lstlisting}[caption={q1.1 on LUBM}, frame=single, breaklines=true, basicstyle=\scriptsize, morekeywords={PREFIX, SELECT, WHERE, UNION, OPTIONAL}, numbers=left, numbersep=4pt, numberstyle=\tiny, label={lst:q1.1_lubm}]
SELECT * WHERE {
{ ?v2 ub:headOf ?v1. } UNION { ?v2 ub:worksFor ?v1. }
?v2 ub:undergraduateDegreeFrom ?v3.
?v4 ub:doctoralDegreeFrom ?v3.
?v5 ub:publicationAuthor ?v2.
{ ?v6 ub:headOf ?v1. } UNION { ?v6 ub:worksFor ?v1. }
{ ?v2 ub:headOf ?v7. } UNION { ?v2 ub:worksFor ?v7. }
<http://www.Department0.University0.edu/UndergraduateStudent91> ub:memberOf ?v1.
?v7 ub:name ?v8. }
\end{lstlisting}

\begin{lstlisting}[caption={q1.2 on LUBM}, frame=single, breaklines=true, basicstyle=\scriptsize, morekeywords={PREFIX, SELECT, WHERE, UNION, OPTIONAL}, numbers=left, numbersep=4pt, numberstyle=\tiny]
SELECT * WHERE {
?v3 ub:emailAddress "UndergraduateStudent91@Department0.University0.edu" .
?v2 ub:emailAddress ?v1 .
OPTIONAL { ?v2 ub:teacherOf ?v4. ?v3 ub:takesCourse ?v4 . } }
\end{lstlisting}

\begin{lstlisting}[caption={q1.3 on LUBM}, label={lst:CP_effective}, frame=single, breaklines=true, basicstyle=\scriptsize, morekeywords={PREFIX, SELECT, WHERE, UNION, OPTIONAL}, numbers=left, numbersep=4pt, numberstyle=\tiny, label={lst:q1.3_lubm}]
SELECT * WHERE {
<http://www.Department1.University0.edu/UndergraduateStudent363> ub:takesCourse ?v1.
OPTIONAL { ?v2 ub:teachingAssistantOf ?v1.
OPTIONAL { ?v2 ub:memberOf ?v3.
?v4 ub:subOrganizationOf ?v3.
?v4 ub:subOrganizationOf ?v5.
?v4 rdf:type ?v6.
OPTIONAL { ?v5 ub:subOrganizationOf ?v7. } } } }
\end{lstlisting}

\begin{lstlisting}[caption={q1.4 on LUBM}, frame=single, breaklines=true, basicstyle=\scriptsize, morekeywords={PREFIX, SELECT, WHERE, UNION, OPTIONAL}, numbers=left, numbersep=4pt, numberstyle=\tiny]
SELECT * WHERE {
?v1	ub:emailAddress	"UndergraduateStudent309@Department12.University0.edu".
OPTIONAL{ ?v1 ub:memberOf ?v2. ?v2 ub:name ?v3.
OPTIONAL{?v5 ub:publicationAuthor ?v4. ?v4 ub:worksFor ?v2.
OPTIONAL{ ?v6 ub:publicationAuthor ?v4. } } } }
\end{lstlisting}

\begin{lstlisting}[caption={q1.5 on LUBM}, frame=single, breaklines=true, basicstyle=\scriptsize, morekeywords={PREFIX, SELECT, WHERE, UNION, OPTIONAL}, numbers=left, numbersep=4pt, numberstyle=\tiny]
SELECT * WHERE {
{ ?v2 <http://www.w3.org/1999/02/22-rdf-syntax-ns#type> ?v3. }
UNION
{ ?v2 ub:name ?v4. }
<http://www.Department0.University0.edu/UndergraduateStudent356> ub:memberOf	?v1.
?v2	ub:worksFor	?v1.
OPTIONAL{ ?v5 ub:advisor ?v2.
OPTIONAL{ ?v5 ub:teachingAssistantOf ?v6.} }
OPTIONAL{ ?v7 ub:advisor ?v2. } }
\end{lstlisting}

\begin{lstlisting}[caption={q1.6 on LUBM}, label={lst:TTCP_effective}, frame=single, breaklines=true, basicstyle=\scriptsize, morekeywords={PREFIX, SELECT, WHERE, UNION, OPTIONAL}, numbers=left, numbersep=4pt, numberstyle=\tiny, label={lst:q1.6_lubm}]
SELECT * WHERE {
?v4 ub:headOf ?v1.
<http://www.Department1.University0.edu/UndergraduateStudent256> ub:memberOf ?v1.
?v3 ub:subOrganizationOf ?v5.
{ ?v2 ub:worksFor ?v1. } UNION { ?v2 ub:headOf ?v1. }
{ ?v2 ub:worksFor ?v3. } UNION { ?v2 ub:headOf ?v3. }
OPTIONAL { ?v6 ub:publicationAuthor ?v2. }
OPTIONAL { { ?v7 ub:headOf ?v1. } UNION { ?v7 ub:worksFor ?v1. } } }
\end{lstlisting}

\begin{lstlisting}[caption={q2.1 on LUBM}, frame=single, breaklines=true, basicstyle=\scriptsize, morekeywords={PREFIX, SELECT, WHERE, UNION, OPTIONAL}, numbers=left, numbersep=4pt, numberstyle=\tiny]
SELECT * WHERE {
{?st ub:teachingAssistantOf ?course.
OPTIONAL { ?st ub:takesCourse ?course2. ?pub1 ub:publicationAuthor ?st. } }
{?prof ub:teacherOf ?course. ?st ub:advisor ?prof.
OPTIONAL { ?prof ub:researchInterest ?resint. ?pub2 ub:publicationAuthor ?prof. } } }
\end{lstlisting}

\begin{lstlisting}[caption={q2.2 on LUBM}, frame=single, breaklines=true, basicstyle=\scriptsize, morekeywords={PREFIX, SELECT, WHERE, UNION, OPTIONAL}, numbers=left, numbersep=4pt, numberstyle=\tiny]
SELECT * WHERE {
{?pub rdf:type ub:Publication. ?pub ub:publicationAuthor ?st. ?pub ub:publicationAuthor ?prof.
OPTIONAL { ?st ub:emailAddress ?ste. ?st ub:telephone ?sttel. } }
{?st ub:undergraduateDegreeFrom ?univ. ?dept ub:subOrganizationOf ?univ.
OPTIONAL { ?head ub:headOf ?dept. ?others ub:worksFor ?dept. } }
{?st ub:memberOf ?dept. ?prof ub:worksFor ?dept.
OPTIONAL { ?prof ub:doctoralDegreeFrom ?univ1. ?prof ub:researchInterest ?resint1. } } }
\end{lstlisting}

\begin{lstlisting}[caption={q2.3 on LUBM}, frame=single, breaklines=true, basicstyle=\scriptsize, morekeywords={PREFIX, SELECT, WHERE, UNION, OPTIONAL}, numbers=left, numbersep=4pt, numberstyle=\tiny]
SELECT * WHERE {
{?pub ub:publicationAuthor ?st. ?pub ub:publicationAuthor ?prof.
?st rdf:type ub:GraduateStudent.
OPTIONAL { ?st ub:undergraduateDegreeFrom ?univ1. ?st ub:telephone ?sttel. } }
{?st ub:advisor ?prof.
OPTIONAL { ?prof ub:doctoralDegreeFrom ?univ. ?prof ub:researchInterest ?resint. } }
{?st ub:memberOf ?dept. ?prof ub:worksFor ?dept. ?prof rdf:type ub:FullProfessor.
OPTIONAL { ?head ub:headOf ?dept. ?others ub:worksFor ?dept. } } }
\end{lstlisting}

\begin{lstlisting}[caption={q2.4 on LUBM}, frame=single, breaklines=true, basicstyle=\scriptsize, morekeywords={PREFIX, SELECT, WHERE, UNION, OPTIONAL}, numbers=left, numbersep=4pt, numberstyle=\tiny]
SELECT * WHERE {
?x ub:worksFor <http://www.Department0.University0.edu>.
?x rdf:type ub:FullProfessor.
OPTIONAL { ?y ub:advisor ?x. ?x ub:teacherOf ?z. ?y ub:takesCourse ?z. } }
\end{lstlisting}

\begin{lstlisting}[caption={q2.5 on LUBM}, frame=single, breaklines=true, basicstyle=\scriptsize, morekeywords={PREFIX, SELECT, WHERE, UNION, OPTIONAL}, numbers=left, numbersep=4pt, numberstyle=\tiny]
SELECT * WHERE {
?x ub:worksFor <http://www.Department0.University12.edu>.
?x rdf:type ub:FullProfessor.
OPTIONAL { ?y ub:advisor ?x. ?x ub:teacherOf ?z. ?y ub:takesCourse ?z. } }
\end{lstlisting}

\begin{lstlisting}[caption={q2.6 on LUBM}, frame=single, breaklines=true, basicstyle=\scriptsize, morekeywords={PREFIX, SELECT, WHERE, UNION, OPTIONAL}, numbers=left, numbersep=4pt, numberstyle=\tiny]
SELECT * WHERE {
?x ub:worksFor <http://www.Department0.University12.edu>.
?x rdf:type ub:FullProfessor.
OPTIONAL { ?x ub:emailAddress ?y1. ?x ub:telephone ?y2. ?x ub:name ?y3. } }
\end{lstlisting}

\subsection{Queries on DBpedia}

\begin{lstlisting}[caption={Prefixes of DBpedia Queries}, frame=single, breaklines=true, basicstyle=\scriptsize, morekeywords={PREFIX, SELECT, WHERE, UNION, OPTIONAL}]
PREFIX rdf:   <http://www.w3.org/1999/02/22-rdf-syntax-ns#>
PREFIX rdfs:  <http://www.w3.org/2000/01/rdf-schema#>
PREFIX foaf:  <http://xmlns.com/foaf/0.1/>
PREFIX purl:  <http://purl.org/dc/terms/>
PREFIX skos:  <http://www.w3.org/2004/02/skos/core#>
PREFIX nsprov:  <http://www.w3.org/ns/prov#>
PREFIX owl:  <http://www.w3.org/2002/07/owl#>
PREFIX dbo:  <http://dbpedia.org/ontology/>
PREFIX dbr:  <http://dbpedia.org/resource/>
PREFIX dbp:   <http://dbpedia.org/property/>
\end{lstlisting}

\begin{lstlisting}[caption={q1.1 on DBpedia}, label={lst:TT_effective}, frame=single, breaklines=true, basicstyle=\scriptsize, morekeywords={PREFIX, SELECT, WHERE, UNION, OPTIONAL}, numbers=left, numbersep=4pt, numberstyle=\tiny, label={lst:q1.1_dbpedia}]
SELECT * WHERE {
{ ?v3 rdfs:label ?v7. } UNION { ?v3 foaf:name ?v7. }
{ ?v1 purl:subject ?v3. } UNION { ?v3 skos:subject ?v1. }
?v3 rdfs:label ?v4.
?v5 nsprov:wasDerivedFrom ?v2.
?v1 owl:sameAs ?v6.
?v1 dbo:wikiPageWikiLink dbr:Economic_system.
?v1 nsprov:wasDerivedFrom ?v2. }
\end{lstlisting}

\begin{lstlisting}[caption={q1.2 on DBpedia}, frame=single, breaklines=true, basicstyle=\scriptsize, morekeywords={PREFIX, SELECT, WHERE, UNION, OPTIONAL}, numbers=left, numbersep=4pt, numberstyle=\tiny]
SELECT * WHERE {
{ ?v3 purl:subject ?v5. OPTIONAL{ ?v5 rdfs:label ?v6 } }
UNION
{ ?v5 skos:subject ?v3. OPTIONAL{ ?v5 foaf:name ?v6 } }
?v1 dbo:wikiPageWikiLink dbr:Economic_system .
?v1 nsprov:wasDerivedFrom ?v2 .
?v3 dbo:wikiPageWikiLink ?v4 .
?v3 nsprov:wasDerivedFrom ?v2 . }
\end{lstlisting}

\begin{lstlisting}[caption={q1.3 on DBpedia}, frame=single, breaklines=true, basicstyle=\scriptsize, morekeywords={PREFIX, SELECT, WHERE, UNION, OPTIONAL}, numbers=left, numbersep=4pt, numberstyle=\tiny]
SELECT * WHERE {
dbr:Air_masses foaf:isPrimaryTopicOf ?v1.
?v2 foaf:isPrimaryTopicOf ?v1.
OPTIONAL {
?v2 dbo:wikiPageRedirects ?v3. ?v4 foaf:primaryTopic ?v2.
OPTIONAL{
?v5 dbo:wikiPageWikiLink ?v3.
OPTIONAL{ ?v6 dbo:wikiPageRedirects ?v5.
OPTIONAL{ ?v6 dbo:wikiPageWikiLink ?v7. } } } } }
\end{lstlisting}

\begin{lstlisting}[caption={q1.4 on DBpedia}, frame=single, breaklines=true, basicstyle=\scriptsize, morekeywords={PREFIX, SELECT, WHERE, UNION, OPTIONAL}, numbers=left, numbersep=4pt, numberstyle=\tiny]
SELECT * WHERE {
dbr:Functional_neuroimaging purl:subject ?v1.
OPTIONAL{
?v1 owl:sameAs ?v2. ?v1 rdf:type ?v3. ?v4 owl:sameAs ?v2. ?v5 skos:related ?v4.
OPTIONAL{ ?v6 skos:related ?v4. }
OPTIONAL{
{ ?v7 purl:subject ?v1. } UNION { ?v1 skos:subject ?v7. }
OPTIONAL{
{ ?v7 purl:subject ?v8. } UNION { ?v8 skos:subject ?v7. } } } } }
\end{lstlisting}

\begin{lstlisting}[caption={q1.5 on DBpedia}, frame=single, breaklines=true, basicstyle=\scriptsize, morekeywords={PREFIX, SELECT, WHERE, UNION, OPTIONAL}, numbers=left, numbersep=4pt, numberstyle=\tiny]
SELECT * WHERE {
{ ?v2 purl:subject ?v3. } UNION { ?v2 dbo:wikiPageWikiLink ?v4. }
?v1 dbo:wikiPageWikiLink dbr:Abdul_Rahim_Wardak.
?v2 dbo:wikiPageWikiLink ?v1.
OPTIONAL{ ?v5 owl:sameAs ?v2.
OPTIONAL{ ?v5 dbo:wikiPageLength ?v6. } }
OPTIONAL{ ?v2 skos:prefLabel ?v7 . } }
\end{lstlisting}

\begin{lstlisting}[caption={q1.6 on DBpedia}, frame=single, breaklines=true, basicstyle=\scriptsize, morekeywords={PREFIX, SELECT, WHERE, UNION, OPTIONAL}, numbers=left, numbersep=4pt, numberstyle=\tiny]
SELECT * WHERE {
{ ?v2 foaf:primaryTopic ?v1. } UNION { ?v1 foaf:isPrimaryTopicOf ?v2. }
{ ?v2 foaf:primaryTopic ?v3. } UNION { ?v3 foaf:isPrimaryTopicOf ?v2. }
?v1 dbo:wikiPageWikiLink dbr:Category:Cell_biology.
?v3 dbo:wikiPageWikiLink ?v1.
OPTIONAL{
{ ?v2 foaf:primaryTopic ?v4. } UNION { ?v4 foaf:isPrimaryTopicOf ?v2. } }
OPTIONAL{ ?v5 dbo:phylum ?v3. ?v6 dbo:phylum ?v3.
OPTIONAL{
{ ?v7 foaf:primaryTopic ?v5. } UNION { ?v5 foaf:isPrimaryTopicOf ?v7. } } } }
\end{lstlisting}

\begin{lstlisting}[caption={q2.1 on DBpedia}, frame=single, breaklines=true, basicstyle=\scriptsize, morekeywords={PREFIX, SELECT, WHERE, UNION, OPTIONAL}, numbers=left, numbersep=4pt, numberstyle=\tiny]
SELECT * WHERE {
{ ?v6 a dbo:PopulatedPlace. ?v6 dbo:abstract ?v1.
?v6 rdfs:label ?v2. ?v6 geo:lat ?v3. ?v6 geo:long ?v4.
OPTIONAL { ?v6 foaf:depiction ?v8. } }
OPTIONAL { ?v6 foaf:homepage ?v10. }
OPTIONAL { ?v6 dbo:populationTotal ?v12. }
OPTIONAL { ?v6 dbo:thumbnail ?v14. } }
\end{lstlisting}

\begin{lstlisting}[caption={q2.2 on DBpedia}, frame=single, breaklines=true, basicstyle=\scriptsize, morekeywords={PREFIX, SELECT, WHERE, UNION, OPTIONAL}, numbers=left, numbersep=4pt, numberstyle=\tiny]
SELECT * WHERE {
?v3 foaf:homepage ?v0. ?v3 a dbo:SoccerPlayer. ?v3 dbp:position ?v6.
?v3 dbp:clubs ?v8. ?v8 dbo:capacity ?v1. ?v3 dbo:birthPlace ?v5.
OPTIONAL { ?v3 dbo:number ?v9. } }
\end{lstlisting}

\begin{lstlisting}[caption={q2.3 on DBpedia}, frame=single, breaklines=true, basicstyle=\scriptsize, morekeywords={PREFIX, SELECT, WHERE, UNION, OPTIONAL}, numbers=left, numbersep=4pt, numberstyle=\tiny]
SELECT * WHERE {
?v5 dbo:thumbnail ?v4. ?v5 rdf:type dbo:Person. ?v5 rdfs:label ?v.
?v5 foaf:homepage ?v8.
OPTIONAL { ?v5 foaf:homepage ?v10. } }
\end{lstlisting}

\begin{lstlisting}[caption={q2.4 on DBpedia}, frame=single, breaklines=true, basicstyle=\scriptsize, morekeywords={PREFIX, SELECT, WHERE, UNION, OPTIONAL}, numbers=left, numbersep=4pt, numberstyle=\tiny]
SELECT * WHERE {
{ ?v2 a dbo:Settlement. ?v2 rdfs:label ?v. ?v6 a dbo:Airport.
?v6 dbo:city ?v2. ?v6 dbp:iata ?v5.
OPTIONAL { ?v6 foaf:homepage ?v7. } }
OPTIONAL { ?v6 dbp:nativename ?v8. } }
\end{lstlisting}

\begin{lstlisting}[caption={q2.5 on DBpedia}, frame=single, breaklines=true, basicstyle=\scriptsize, morekeywords={PREFIX, SELECT, WHERE, UNION, OPTIONAL}, numbers=left, numbersep=4pt, numberstyle=\tiny]
SELECT * WHERE {
?v4 skos:subject ?v. ?v4 foaf:name ?v6.
OPTIONAL { ?v4 rdfs:comment ?v8. } }
\end{lstlisting}

\begin{lstlisting}[caption={q2.6 on DBpedia}, frame=single, breaklines=true, basicstyle=\scriptsize, morekeywords={PREFIX, SELECT, WHERE, UNION, OPTIONAL}, numbers=left, numbersep=4pt, numberstyle=\tiny]
SELECT * WHERE {
?v0 rdfs:comment ?v1. ?v0 foaf:page ?v.
OPTIONAL { ?v0 skos:subject ?v6. }
OPTIONAL { ?v0 dbp:industry ?v5. }
OPTIONAL { ?v0 dbp:location ?v2. }
OPTIONAL { ?v0 dbp:locationCountry ?v3. }
OPTIONAL { ?v0 dbp:locationCity ?v9. ?a dbp:manufacturer ?v0. }
OPTIONAL { ?v0 dbp:products ?v11. ?b dbp:model ?v0. }
OPTIONAL { ?v0 georss:point ?v10. }
OPTIONAL { ?v0 rdf:type ?v7. } }
\end{lstlisting}

%% file: template.bbl
\begin{thebibliography}{10}
\providecommand{\url}[1]{{#1}}
\providecommand{\urlprefix}{URL }
\expandafter\ifx\csname urlstyle\endcsname\relax
  \providecommand{\doi}[1]{DOI~\discretionary{}{}{}#1}\else
  \providecommand{\doi}{DOI~\discretionary{}{}{}\begingroup
  \urlstyle{rm}\Url}\fi

\bibitem{Dbpediaurl}
Dbpedia.
\newblock \urlprefix\url{https://wiki.dbpedia.org/}

\bibitem{lubmurl}
Lubm.
\newblock \urlprefix\url{http://swat.cse.lehigh.edu/projects/lubm/}

\bibitem{virtuosourl}
Virtuoso.
\newblock \urlprefix\url{https://virtuoso.openlinksw.com/}

\bibitem{DBLP:journals/vldb/AbadiMMH09}
Abadi, D.J., Marcus, A., Madden, S., Hollenbach, K.: Sw-store: a vertically
  partitioned {DBMS} for semantic web data management.
\newblock {VLDB} J. \textbf{18}(2), 385--406 (2009)

\bibitem{vldb07_Abadi:2007}
Abadi, D.J., Marcus, A., Madden, S., Hollenbach, K.J.: Scalable semantic web
  data management using vertical partitioning.
\newblock In: VLDB, pp. 411--422. {ACM} (2007)

\bibitem{al2017optimizing}
Al-Kateb, M., Sinclair, P., Crolotte, A., Ma, L., Au, G., Nair, S.: Optimizing
  union all join queries in teradata.
\newblock In: 2017 IEEE 33rd International Conference on Data Engineering
  (ICDE), pp. 1209--1212. IEEE (2017)

\bibitem{DBLP:conf/sigmod/Atre15}
Atre, M.: \emph{Left Bit Right}: For {SPARQL} join queries with {OPTIONAL}
  patterns (left-outer-joins).
\newblock In: SIGMOD, pp. 1793--1808. {ACM} (2015)

\bibitem{DBLP:journals/jacm/BernsteinC81}
Bernstein, P.A., Chiu, D.W.: Using semi-joins to solve relational queries.
\newblock J. {ACM} \textbf{28}(1), 25--40 (1981)

\bibitem{bonifati2020analytical}
Bonifati, A., Martens, W., Timm, T.: An analytical study of large sparql query
  logs.
\newblock The VLDB Journal \textbf{29}(2), 655--679 (2020)

\bibitem{DBLP:conf/sigmod/BorneaDKSDUB13}
Bornea, M.A., Dolby, J., Kementsietsidis, A., Srinivas, K., Dantressangle, P.,
  Udrea, O., Bhattacharjee, B.: Building an efficient {RDF} store over a
  relational database.
\newblock In: SIGMOD, pp. 121--132. {ACM} (2013)

\bibitem{chebotko2009semantics}
Chebotko, A., Lu, S., Fotouhi, F.: Semantics preserving sparql-to-sql
  translation.
\newblock Data \& Knowledge Engineering \textbf{68}(10), 973--1000 (2009)

\bibitem{10.1007/978-3-030-30793-6_15}
Hogan, A., Riveros, C., Rojas, C., Soto, A.: A worst-case optimal join
  algorithm for sparql.
\newblock In: C.~Ghidini, O.~Hartig, M.~Maleshkova, V.~Sv{\'a}tek, I.~Cruz,
  A.~Hogan, J.~Song, M.~Lefran{\c{c}}ois, F.~Gandon (eds.) The Semantic Web --
  ISWC 2019, pp. 258--275. Springer International Publishing, Cham (2019)

\bibitem{DBLP:journals/semweb/LehmannIJJKMHMK15}
Lehmann, J., Isele, R., Jakob, M., Jentzsch, A., Kontokostas, D., Mendes, P.N.,
  Hellmann, S., Morsey, M., van Kleef, P., Auer, S., Bizer, C.: Dbpedia - {A}
  large-scale, multilingual knowledge base extracted from wikipedia.
\newblock Semantic Web \textbf{6}(2), 167--195 (2015)

\bibitem{zouSPARQLUO:22}
Lei~Zou Yue~Pang, M.T.{\"{O}}., Chen, J.: Efficient execution of sparql queries
  with optional and union expressions (full version).
\newblock Tech. rep. (2022).
\newblock
  \urlprefix\url{https://anonymous.4open.science/r/gStore-UO/full_version.pdf}

\bibitem{10.1145/2500130}
Letelier, A., P\'{e}rez, J., Pichler, R., Skritek, S.: Static analysis and
  optimization of semantic web queries.
\newblock ACM Trans. Database Syst. \textbf{38}(4) (2013).
\newblock \doi{10.1145/2500130}.
\newblock \urlprefix\url{https://doi.org/10.1145/2500130}

\bibitem{DBLP:journals/pvldb/MhedhbiS19}
Mhedhbi, A., Salihoglu, S.: Optimizing subgraph queries by combining binary and
  worst-case optimal joins.
\newblock Proc. {VLDB} Endow. \textbf{12}(11), 1692--1704 (2019)

\bibitem{Neumann2009}
Neumann, T., Weikum, G.: The {RDF-3X} engine for scalable management of {RDF}
  data.
\newblock VLDB Journal \textbf{19}(1), 91--113 (2009)

\bibitem{DBLP:journals/tods/PerezAG09}
P{\'{e}}rez, J., Arenas, M., Guti{\'{e}}rrez, C.: Semantics and complexity of
  {SPARQL}.
\newblock {ACM} Trans. Database Syst. \textbf{34}(3), 16:1--16:45 (2009)

\bibitem{prud'hommeaux_bertails_2008}
Prud'hommeaux, E., Bertails, A.: A mapping of sparql onto conventional sql
  (2008).
\newblock
  \urlprefix\url{https://www.w3.org/2008/07/MappingRules/StemMapping#sqlOpt}

\bibitem{DBLP:conf/sigmod/RaoPZ04}
Rao, J., Pirahesh, H., Zuzarte, C.: Canonical abstraction for outerjoin
  optimization.
\newblock In: SIGMOD, pp. 671--682. {ACM} (2004)

\bibitem{10.1145/3178876.3186003}
Stefanoni, G., Motik, B., Kostylev, E.V.: Estimating the cardinality of
  conjunctive queries over rdf data using graph summarisation.
\newblock In: Proceedings of the 2018 World Wide Web Conference, WWW '18, p.
  1043–1052. International World Wide Web Conferences Steering Committee,
  Republic and Canton of Geneva, CHE (2018).
\newblock \doi{10.1145/3178876.3186003}.
\newblock \urlprefix\url{https://doi.org/10.1145/3178876.3186003}

\bibitem{Wilkinson:Jena2}
Wilkinson, K., Sayers, C., Kuno, H.A., Reynolds, D.: Efficient \uppercase{RDF}
  storage and retrieval in jena2.
\newblock In: The first International Workshop on Semantic Web and Databases,
  pp. 131--150 (2003)

\bibitem{DBLP:journals/pvldb/YuanLWJZL13}
Yuan, P., Liu, P., Wu, B., Jin, H., Zhang, W., Liu, L.: Triplebit: a fast and
  compact system for large scale {RDF} data.
\newblock Proc. {VLDB} Endow. \textbf{6}(7), 517--528 (2013)

\bibitem{DBLP:journals/pvldb/ZouMCOZ11}
Zou, L., Mo, J., Chen, L., {\"{O}}zsu, M.T., Zhao, D.: gstore: Answering
  {SPARQL} queries via subgraph matching.
\newblock Proc. {VLDB} Endow. \textbf{4}(8), 482--493 (2011)

\end{thebibliography}
